\documentclass[oneside,english]{amsart}
\usepackage[T1]{fontenc}
\usepackage[latin9]{inputenc}
\usepackage{geometry}
\usepackage{float}
\usepackage{wrapfig}
\usepackage{units}
\usepackage{bm}
\usepackage{amstext}
\usepackage{amsthm}
\usepackage{amssymb}
\usepackage{graphicx}

\makeatletter

\floatstyle{ruled}
\newfloat{algorithm}{tbp}{loa}
\providecommand{\algorithmname}{Algorithm}
\floatname{algorithm}{\protect\algorithmname}

\numberwithin{equation}{section}
\numberwithin{figure}{section}
\theoremstyle{plain}
\newtheorem{thm}{\protect\theoremname}
\theoremstyle{remark}
\newtheorem{rem}[thm]{\protect\remarkname}
\theoremstyle{definition}
\newtheorem{example}[thm]{\protect\examplename}
\theoremstyle{plain}
\newtheorem{lem}[thm]{\protect\lemmaname}
\theoremstyle{plain}
\newtheorem{cor}[thm]{\protect\corollaryname}

\usepackage{stmaryrd}
\usepackage{cite}
\usepackage{algorithm}
\usepackage[noend]{algpseudocode}

\def\blfootnote{\xdef\@thefnmark{}\@footnotetext}

\@ifundefined{showcaptionsetup}{}{%
 \PassOptionsToPackage{caption=false}{subfig}}
\usepackage{subfig}
\makeatother

\usepackage{babel}
\providecommand{\corollaryname}{Corollary}
\providecommand{\examplename}{Example}
\providecommand{\lemmaname}{Lemma}
\providecommand{\remarkname}{Remark}
\providecommand{\theoremname}{Theorem}

\begin{document}

\title{Dimension Reduction for Origin-Destination Flow Estimation:\\
 Blind Estimation Made Possible}

\author{Jingyuan Xia, Wei Dai, John Polak, and Michel Bierlaire }
\begin{abstract}
This paper studies the problem of estimating origin-destination (OD)
flows from link flows. As the number of link flows is typically much
less than that of OD flows, the inverse problem is severely ill-posed
and hence prior information is required to recover the ground truth.
The basic approach in the literature relies on a forward model where
the so called traffic assignment matrix maps OD flows to link flows.
Due to the ill-posedness of the problem, prior information on the
assignment matrix and OD flows are typically needed. 

The main contributions of this paper include a dimension reduction
of the inquired flows from $O\left(n^{2}\right)$ to $O\left(n\right)$,
and a demonstration that for the first time the ground truth OD flows
can be uniquely identified with no or little prior information. To
cope with the ill-posedness due to the large number of unknowns, a
new forward model is developed which does not involve OD flows directly
but is built upon the flows characterized only by their origins, henceforth
referred as \emph{O-flows}. The new model preserves all the OD information
and more importantly reduces the dimension of the inverse problem
substantially. A Gauss-Seidel method is deployed to solve the inverse
problem, and a necessary condition for the uniqueness of the solution
is proved. Simulations demonstrate that blind estimation where no
prior information is available is possible for some network settings.
Some challenging network settings are identified and discussed, where
a remedy based on temporal patterns of the O-flows is developed and
numerically shown effective.
\end{abstract} 

\maketitle
\blfootnote{
	Jingyuan Xia and Wei Dai <j.xia16 and wei.dai1@imperial.ac.uk>, Department of Electric and Electronic Engineering, Imperial college London, London, SW7 2AZ, UK.
	
	John Polak <j.polak@imperial.ac.uk>, Department of Civil and Environmental Engineering, Imperial college London, London, SW7 2AZ, UK.
	
	Michel Bierlaire <michel.bierlaire@epfl.ch>, Transport and Mobility Laboratory, Ecole Polytechnique Fédérale de Lausanne, Lausanne, Switzerland.
	
	The authors would thank Mr. Maxime L Ferreira Da Costa at Imperial College London for the discussions during the initial development of the technical approach, and Dr. Gunna Flotterod at Royal Institute of Technology for the discussions and information on the literature in OD estimation.}

\section{\label{sec:Introduction}Introduction}

The origin-destination (OD) flow estimation problem can be stated
as follows. Consider a network $\mathcal{G}\left(\mathcal{N},\mathcal{L}\right)$
specified by the set of nodes $\mathcal{N}$ (also known as vertices)
and links $\mathcal{L}$ (also called edges). Suppose that the quantities
of link flows are given during the time horizon involving multiple
consecutive sampling time intervals. The task is to estimate the volume
of the OD flows.

OD flow estimation is essential to many network analysis tasks. This
paper focuses on transportation networks though the developed approach
can be extended to other types of networks. In the transportation
community, it has been widely accepted that OD information reflects
the travel demands and plays an essential role in long term infrastructure
planning, traffic prediction under unexpected changes in the infrastructure
or the traffic status, and commercial applications linked to population
migration \cite{de2011modelling}, and serves as an important input
to traffic simulation models \cite{toledo2003calibration,bauer2017quasi}.
Despite its importance, accurate OD information is difficult, expensive,
and sometimes impossible to be obtained. The old fashioned household
survey data are expensive and time-consuming to collect and typically
incomplete and biased \cite{willumsen1978estimation,cascetta1984estimation}.
Commercial transportation service data from taxi and Uber are highly
biased towards commercial activities. Recently modern technologies,
for instance GPS, mobile phone, automatic plate recognition systems,
automatic vehicle identification systems, and combinations of these
systems, provide new data sources and new opportunities for acquiring
OD information \cite{antoniou2011synthesis,parry2012estimation,jin2014location,alexander2015origin,moreira2016time,yang2017origin}.
However, the data are highly privacy-sensitive and sometimes error-prone.

The problem of interest is to infer OD flows from the link flows.
It is ill-posed as the number of link flows (observations) is typically
much less than that of OD flows (unknowns). Let $n_{n}$ be the number
of the nodes in the network and $n_{\ell}$ be the number of the links.
Let $n_{\ell}=cn_{n}$ where $c\in\mathbb{R}^{+}$ denotes the average
number of links per node. In a transportation network, it is typical
that $2\le c\le4$. At the same time, the number of OD flows can be
as large as $n_{n}^{2}$. Even when restricting the distance between
the origin and destination nodes, the number of possible OD flows
can be still much larger than that of link flows. The inverse problem
typically does not admit a unique solution.

In the literature, a common approach is based on a linear forward
model that maps OD flows to link flows. Refer to the linear operator
that maps OD flows to link flows as traffic assignment \cite{willumsen1978estimation,cascetta1984estimation,de2011modelling}
(also known as link choice proportions \cite{lo1999decomposition,lo1996estimation},
or path proportions \cite{frederix2013dynamic,cao2000scalable}).
It gives the fraction of an OD flow passing through a specific link.
The assignment matrix based approach requires prior information on
both the assignment matrix and the OD flows: it is typically assumed
that either the assignment matrix or its approximation is given based
on historical data or traffic modeling; prior information on the inquired
OD flows includes historical data, statistical models, temporal and/or
spatial relationship among the flows, etc. Beyond the linear model,
nonlinear models such as user equilibrium model \cite{bar2010traffic,shen2012new,frederix2013dynamic,lu2013dynamic}
and network loading model \cite{balakrishna2007offline,balakrishna2008incorporating,lu2013dynamic}
have been deployed to describe the complex relationships between OD
flows and link flows. Prior knowledge in terms of domain knowledge
or historical data is also needed in addressing the ill-posedness
of the OD flow estimation problem. 

In this paper, the fundamental question of interest is whether the
ill-posedness of the inverse problem can be addressed with no or little
prior information, i.e., whether blind estimation is possible. This
paper adopts the assignment matrix based linear model for simplicity
and the purpose of proof-of-concept. We assume that the traffic assignment
matrix is unknown but fixed during the whole time horizon. The OD
flows are dynamic, meaning that the OD flows vary across different
sampling time intervals. 

The main contributions of this paper are summarized in the following. 
\begin{itemize}
\item A new linear model is developed to allow substantial dimension reduction
of the inquired flows from $O\left(n^{2}\right)$ to $O\left(n\right)$.
More specifically, define a flow originating from a node as an O-flow.
A linear model is constructed to map O-flows to link flows. With a
slight abuse of terminology, refer to the corresponding linear operator
also as traffic assignment matrix. (The two different assignment matrices,
one corresponding to OD flows and the other O-flows, can be distinguished
according to the context.) The O-flows together with the corresponding
assignment matrix preserve the OD flow information. At the same time,
the number of O-flows is at most the number of the nodes in the network,
which is typically less than the number of links and much less than
the number of OD flows. 
\item It is numerically demonstrated that for the first time the ground
truth OD flows can be uniquely identified without any prior information
of either the flows or the traffic assignment. In this paper, the
OD flow information is inferred by jointly estimating both the O-flows
and the corresponding assignment matrix. An iterative algorithm is
developed to solve the joint estimation problem based on the Gauss-Seidel
method. Simulations in Section \ref{sec:Simulations} show that the
ground truth OD flows can be estimated with high accuracy for bidirectional
networks. In one tested scenario, 1660 OD flows are accurately estimated
from 224 link flows without any prior information. According to the
authors' knowledge, no similar result has been reported before in
the literature. 
\item A necessary condition is derived for the uniqueness of the solution
of the O-flow model. Based on this necessary condition, we show that
in general unidirectional networks do not admit unique solutions.
A remedy is proposed to promote a unique solution by assuming temporal
patterns in O-flows, in particular that the coefficients of the discrete
cosine transform (DCT) of O-flows have only a few significant components.
Numerical simulations have demonstrated the effectiveness of this
remedy. Here, the only prior information used is that O-flows are
sparse in the DCT transform domain. There is no need for historical
data, or the knowledge which DCT coefficients are nonzero. 
\end{itemize}
In summary, blind estimation is made possible thanks to the substantial
dimension reduction of the new linear model. 

As a starting point, this paper focuses on a simple linear model involving
a static assignment matrix. In reality, the relationship between OD
flows and link flows is nonlinear according to the fundamental diagram
of traffic flow\footnote{One way to handle nonlinearity is to use local linear approximations
derived from Taylor series. Due to the fundamental nature of linear
models, they are the focus of this paper.}. Even in the local linear approximation regime, the assignment matrix
can be dynamic in real life situations. Our new model and approach
can be adapted and extended to address more complicated scenarios,
which we leave as future work. 

This paper is organized as follows. Section \ref{sec:Models} reviews
popular models and techniques in the literature briefly. Section \ref{sec:NewModel}
introduces our linear model used for OD flow estimation. Section \ref{sec:Blind-Estimation}
describes the computational procedure, derives a necessary condition
for the uniqueness of the solution, and analyzes unidirectional networks.
Simulation results are presented in Section \ref{sec:Simulations}
to demonstrate the feasibility for blind estimation. Conclusions and
possible future work are given in Section \ref{sec:Conclusion}.

\section{\label{sec:Models}Models for OD Estimation in literature}

Directed graphs $\mathcal{G}\left(\mathcal{N},\mathcal{L}\right)$
are considered in this paper. Denote the traffic flow cont of the
link from a node $i$ to its adjacent node $j$ by $y_{ij}$. Similarly
the traffic flow count from an origin node $o$ to a destination node
$d$ is denoted by $s_{od}$. Assume that the network has $n_{n}$
many nodes, $n_{\ell}$ many links, and $n_{{\rm OD}}$ many OD pairs. 

\subsection{\label{subsec:Forward-Models}Forward Models}

Group the link flows $y_{ij}$ into a vector $\bm{y}$. It is clear
that $\bm{y}\in\bar{\mathbb{R}}^{n_{\ell}}$ where $\bar{\mathbb{R}}:=\mathbb{R}^{+}\bigcup\left\{ 0\right\} $
is the set of non-negative real numbers\footnote{For modeling and computational simplicity, we relax the domain of
traffic counts from $\mathbb{Z}^{+}\bigcup\left\{ 0\right\} $ to
$\mathbb{R}^{+}\bigcup\left\{ 0\right\} $.}. Let $\bm{s}\in\bar{\mathbb{R}}^{n_{{\rm OD}}}$ denote the OD flow
vector. The most popular linear model in the literature \cite{willumsen1981simplified,cascetta1984estimation,bell1991estimation,hazelton2001inference,bierlaire2004efficient,xie2010maximum}
is given by 
\begin{equation}
\bm{y}=\bm{A}\bm{s},\label{eq:OD2Link}
\end{equation}
where the matrix $\bm{A}$ is referred to as traffic assignment matrix,
and its entries $a_{ij,od}$ gives the fraction of OD flow $s_{od}$
passing the link from $i$ to $j$. 

Finer time resolution can be added to the above model. Refer to the
time window under which the link flow data are collected as sampling
time interval, and the time period of all the consecutive sampling
time intervals as time horizon. Model (\ref{eq:OD2Link}) implicitly
assumes that all the OD flows finish in one sampling time interval,
which is reasonable in real life when the sampling time interval is
sufficiently long, for example, a day long. With finer time resolution,
for example the sampling time interval is of ten minutes long, multiple
sampling time intervals may be involved during the lifetime of OD
flows. The linear relations from OD flows to link flows are then described
by a convolution form: 
\begin{align}
 & \bm{y}^{t}=\sum_{\tau=1}^{\tau_{\max}}\bm{A}^{\tau}\bm{s}^{t-\tau+1},\;\mbox{or equivalently }\bm{y}^{t}=\bm{A}^{t}*\bm{s}^{t},\label{eq:OD2Link-MultiStep}
\end{align}
where $\bm{y}^{t},$ $\bm{A}^{t}$, and $\bm{s}^{t}$ are the link
flow, traffic assignment matrix, and the OD flow at time interval
$t$ respectively, $\tau_{\max}\in\mathbb{Z}^{+}$ is the maximum
number of consecutive sampling time intervals involved for OD trips,
and the symbol $*$ denotes a convolution which is commonly used in
signal processing community. Refer to this model as multi-step model.
On one hand, it provides finer time resolution. On the other hand,
it involves more unknown variables and more prior information is needed
to solve the inverse problem. 

It is noteworthy that both linear models (\ref{eq:OD2Link}) and (\ref{eq:OD2Link-MultiStep})
assume that the traffic assignment matrix $\bm{A}$ is independent
of the OD flows $\bm{s}$. This model is often referred to as a separable
model in the literature. In real life scenarios, the assignment matrix
$\bm{A}$ and the OD flows are correlated (non-separable) according
to the fundamental diagram of traffic flow. A non-separable model
results in a nonlinear model, e.g., user equilibrium model \cite{willumsen1981simplified,bar2010traffic,shen2012new,frederix2013dynamic,lu2013dynamic},
and goes beyond the scope of this paper. 

\subsection{\label{subsec:OD-Estimation-Literature}OD Estimation: Solving the
Inverse Problem}

The inverse problem of OD estimation is severely ill-posed because
the number of OD flows (unknown variables) can be much larger than
that of observed link flows. To address this issue, prior information
on both the assignment matrix and the OD flows is necessary. Based
on different assumptions about the prior information, different techniques
have been developed to solve the inverse problem. In this following,
we shall discuss some representative approaches. As OD flow estimation
has been an active research topic for many decades, we can only include
a small subset of the literature below. 

In the works \cite{willumsen1978estimation,willumsen1981simplified},
gravity models and entropy maximizing principle have been used to
choose one solution (of the inverse problem) from all feasible ones.
The basic form of gravity models is to approximate the OD flow $s_{od}$
by 
\[
s_{od}=bR_{o}R_{d}c_{od}^{-r},
\]
where $b$ and $r$ are parameters for calibration, $R_{o}$ and $R_{d}$
represent information (e.g. population, employment, the mean income
of the residence) of the origin and the destination respectively,
and $c_{od}$ is the cost of traveling from $o$ to $d$. In entropy
maximizing models \cite{willumsen1978estimation,willumsen1981simplified,lam1991estimation,xie2010maximum},
among all feasible solutions satisfying $\bm{y}=\bm{A}\bm{s}$, the
one that maximizes the entropy function $-\sum_{o,d}\left(s_{od}\ln s_{od}-s_{od}\right)$
is of interest. A variation of this is given by the information minimizing
model \cite{willumsen1981simplified,lam1991estimation}, of which
the solution has a very similar form to that of the entropy maximizing
model. 

When historical data of the OD flows are available, a popular approach
is the Generalized Least Squares (GLS) estimator. Assume that both
the link flow measurement errors and OD flow approximation errors
can be modeled by using multivariate normal distribution. The GLS
estimator is given by \cite{cascetta1984estimation,lo1996estimation,lo1999decomposition}

\begin{equation}
\begin{array}{c}
\min_{\bm{s}}\;\left(\bm{A}\bm{s}-\bm{y}\right)^{T}\bm{W}^{-1}\left(\bm{A}\bm{s}-\bm{y}\right)\end{array}+\left(\bm{s}-\bm{s}^{H}\right)^{T}\bm{V}^{-1}\left(\bm{s}-\bm{s}^{H}\right),\label{eq:GLS OD estimator old}
\end{equation}
where $\bm{W}$ and $\bm{V}$ are the covariance matrices and assumed
to be known a priori. The multivariate normal distribution may result
in negative values of OD flows. One way to address this is to add
a constraint that the inquired OD flows must be non-negative \cite{bell1991estimation}.
Another way is to replace the Gaussian distribution with the Poisson
distribution \cite{cascetta1988unified,lo1996estimation,lo1999decomposition,tebaldi1998bayesian}
which is non-negative, describes the traffic behavior better, but
is computationally more costly. As a fundamental framework, GLS has
been adopted and adapted in many other works, e.g., \cite{menon2015fine,yang2017origin,tympakianaki2018robust}. 

To design and manage modern intelligent transportation systems, it
is important to model the dynamic nature of the OD flows where the
sampling time interval is much less than a day. The dynamics of the
OD flows further increase the number of unknowns and calls for extra
prior information to cope with the ill-posedness. In \cite{ashok2000alternative,bierlaire2004efficient},
the dynamics of the OD flows is modeled by an auto-regressive process.
Let $\partial\bm{s}^{t}=\bm{s}^{t}-\bm{s}^{t,H}$ be the deviations
of OD flows $\bm{s}^{t}$ from the historical data $\bm{s}^{t,H}$.
The auto-regressive model assumes 
\[
\partial\bm{s}^{t}=\sum_{\tau=1}^{q}\bm{B}^{\tau}\partial\bm{s}^{t-\tau}+\bm{w}^{t},
\]
where $\bm{B}^{\tau}\in\mathbb{R}^{n_{OD}\times n_{OD}}$ represents
the linear contribution of $\partial\bm{s}^{t-\tau}$ to $\partial\bm{s}^{t}$
and are assumed to be known a priori, and $\bm{w}_{t}$ denotes the
errors. The multi-step linear model (\ref{eq:OD2Link-MultiStep})
can be equivalently written as 
\[
\partial\bm{y}^{t}=\sum_{\tau=1}^{\tau_{\max}}\bm{A}^{\tau}\partial\bm{s}^{t-\tau+1}+\bm{v}^{t},
\]
where $\partial\bm{y}^{t}=\bm{y}^{t}-\sum_{\tau=1}^{\tau_{\max}}\bm{A}^{\tau}\bm{s}^{t-\tau+1,H}$
and $\bm{v}_{t}$ describes the errors. Assume that $\bm{w}^{t}$
and $\bm{v}^{t}$ are multivariate normal distributed with mean zero
and covariance matrices $\bm{W}$ and $\bm{V}$ respectively. A GLS
estimator similar to (\ref{eq:GLS OD estimator old}) can be then
applied. Based on the auto-regressive modeling, a principal component
analysis has been applied to $\bm{s}^{t}$ for the purpose of further
dimension reduction \cite{prakash2017reducing}. Another way to explore
the temporal correlations of OD flows is to pool identical time periods
over days from the same day category \cite{bauer2017quasi}. 

Spatial structures of the OD flows can be also used to mitigate the
ill-posedness. In stead of considering traffic flows between any pair
of nodes, one can simplify the network and the analysis by aggregating
nodes into zones represented by virtual nodes, connecting them by
virtual links, and analyzing the traffic flows between pairs of the
virtual nodes. This conventional traffic modeling has been adopted
by many works, e.g. \cite{lo1996estimation,lo1999decomposition,cipriani2011gradient},
to cope with limited data and/or computational power. It is typically
left as decisions for domain experts to decide the zone construction,
including the number of zones, the position of the virtual nodes,
and the virtual links connecting virtual nodes. In recent work \cite{menon2015fine},
two techniques are used to reduce the spatial complexity of the network:
automatic zoning and sparsity regularization. Instead of grouping
multiple nodes, the aim of automatic zoning is to select individual
nodes as centers of traffic analysis zones in order to find a Pareto
optimal point for the bi-criteria objective 
\[
\underset{\bm{b}\in\left\{ 0,1\right\} ^{n_{n}}}{\min}\;\left\{ \left\Vert \bm{b}\right\Vert _{0},\left\Vert \bm{y}-\bm{A}^{\left(\bm{b}\right)}\bm{s}^{\left(\bm{b}\right)}\right\Vert _{2}^{2}\right\} ,
\]
where the entries of $\bm{b}$ indicate which nodes are chosen and
which are not, the pseudo-norm $\left\Vert \cdot\right\Vert _{0}$
counts the number of nonzero elements, and $\bm{A}^{\left(\bm{b}\right)}$
and $\bm{s}^{\left(\bm{b}\right)}$ represent truncation of $\bm{A}$
and $\bm{s}$, respectively, based on the nonzero elements of $\bm{b}$.
Under some assumptions, this bi-criteria objective leads to a constrained
nonconvex optimization formulation for automatic zoning. A heuristic
algorithm is also developed. After automatic zoning, the authors further
assume that ``for all but the coarsest of zonings'', OD flows should
be sparse in the sense that most OD flows are so small to be safely
approximated by zero. An $\ell_{1}$-regularization term is added
to the GLS estimator to promote the sparsity of the estimated OD flows,
resulting in 
\[
\underset{\bm{s}\ge\bm{0}}{\min}\;\left(\bm{y}-\bm{A}\bm{s}\right)^{T}\bm{W}^{-1}\left(\bm{y}-\bm{A}\bm{s}\right)+\lambda\left\Vert \bm{s}\right\Vert _{1},
\]
where $\bm{A}$ and $\bm{s}$ are the assignment matrix and OD flows
after zoning, and $\lambda\in\mathbb{R}^{+}$ is a properly chosen
constant to balance data fidelity and solution sparsity. 

It can be observed that all the approaches discussed above heavily
rely on the assumption that either the assignment matrix or its approximation
is known a priori. In reality, such information has to come from somewhere.
The simplest approach is the so called ``all-or-nothing'' assignment
\cite[p. 153, and references therein]{menon2015fine} where only one
path is for one OD trip and the path is chosen to be the one with
the least travel cost (distance or average travel time). Similarly,
one can also allow multiple paths for OD trips and assign a probability
for these paths based on either historical data or a cost measure.
The difficulties \cite{willumsen1981simplified} include, but are
not limited to, the availability/sufficiency of historical data, the
cost measure to choose, possibly different perceptions and objectives
from different drivers in different situations, imperfect knowledge
of the alternative routes, and the dynamics of the assignment matrix
under different road/traffic conditions. The last difficulty can be
addressed by local linear approximation of the assignment matrix \cite{toledo2013estimation,cantelmo2014adaptive}
or a user equilibrium based assignment. However, these techniques
result in a nonlinear model that relies on knowledge of the OD flows,
and hence goes beyond the scope of this paper. 

\section{\label{sec:NewModel}A New Forward Model for OD Flow Estimation }

In the light of the above, the main technical difficulty of using
the standard linear models (\ref{eq:OD2Link},\ref{eq:OD2Link-MultiStep})
for OD flow estimation comes from the large dimension of the unknowns.
The focus of this section is to present a new forward model which
reduces the dimension of the inquired flows from $O\left(n_{n}^{2}\right)$
to $O\left(n_{n}\right)$. 

\subsection{\label{subsec:O-flow-Models}O-flow Based Models}

We build a linear forward model not directly involving the inquired
OD flows. It is built on the traffic flows that are specified with
their origins but not their destinations, henceforth referred to as
\emph{O-flow}. Let $x_{o}$ denote the flow originating from the node
$o$, i.e., $x_{o}=\sum_{d\ne o}s_{od}$. Define the O-flow vector
as $\bm{x}=\left[x_{1},\cdots,x_{o},\cdots,x_{n_{O}}\right]^{T}$,
where $n_{O}$ is the number of valid origins in the network. Let
$p_{ij,o}$ denote the proportion of the O-flow $x_{o}$ that passes
the link from $i$ to $j$. Denote the traffic assignment matrix associated
with O-flows $\bm{x}$ by $\bm{P}$ of which the entries are $p_{ij,o}$.
Assume static traffic assignment and dynamic flows. When the sampling
time interval is longer than the trip time, one has the following
single step model 
\begin{equation}
\bm{y}^{t}=\bm{P}\bm{x}^{t},\quad t=1,2,\cdots,n_{T},\label{eq:O2Link-SingleStep}
\end{equation}
where $n_{T}\in\mathbb{Z}^{+}$ denotes the number of sampling time
intervals involved in the time horizon. Otherwise, one has a multi-step
model where 
\begin{equation}
\bm{y}^{t}=\sum_{\tau=1}^{\tau_{\max}}\bm{P}^{\tau}\bm{x}^{t-\tau+1},\;\mbox{{\rm or }}\bm{y}^{t}=\bm{P}^{t}*\bm{x}^{t},\quad t=1,2,\cdots,n_{T},\label{eq:O2Link-MultiStep}
\end{equation}
where $\tau_{\max}$ is the maximum trip time. 
\begin{rem}
In the sequel, we will sometimes use single step models for illustrate
simplicity. The simulations are based on the multi-step model (\ref{eq:O2Link-MultiStep}). 
\end{rem}

It is clear that the number of O-flows $n_{O}$ is upper bounded by
$n_{n}$. In typical transportation networks, the average number of
links per node, denoted by $c\in\mathbb{R}^{+}$, is larger than 1.
In this case, the number of equations $n_{\ell}=cn_{n}$ is more than
the number of unknown O-flows. The inverse problem is then well-posed
when the assignment matrix is given and of full column rank (which
is possible only when $n_{\ell}\ge n_{O}$). 

The new model does not involve OD flows directly but preserves the
OD flow information. 
\begin{thm}
\label{thm:OD-info-preserved}All the OD flows can be inferred from
the given O-flows and the corresponding assignment matrices. 
\end{thm}

\begin{proof}
In the single step model (\ref{eq:O2Link-SingleStep}), the OD flow
$s_{od}$, $o\ne d$, can be calculated from 
\begin{equation}
s_{od}=x_{o}\left(\sum_{i}p_{id,o}-\sum_{j}p_{dj,o}\right),\label{eq:OFlow2ODFlow-single}
\end{equation}
where the term $x_{o}\sum_{i}p_{id,o}$ calculates the inflow to the
node $d$ that originates from the node $o$, the term $x_{o}\sum_{j}p_{dj,o}$
gives the outflow from the node $d$ that originates from the node
$o$, and the difference between them is clearly the flow ending at
the node $d$ and originating from the node $o$.

Similar arguments can be applied to the multi-step model (\ref{eq:O2Link-SingleStep}),
resulting in
\begin{equation}
s_{od}^{t}=x_{o}^{t}\left(\sum_{\tau=1}^{\tau_{\max}}\sum_{i}p_{id,o}^{\tau}-\sum_{\tau=1}^{\tau_{\max}}\sum_{j}p_{dj,o}^{\tau}\right).\label{eq:OFlow2ODFlow-multiple}
\end{equation}
\end{proof}
\vspace{0cm}

\begin{rem}
The concept of O-flows has been mentioned in the literature, e.g.,
\cite{bauer2017quasi}. However, when coming to OD flow estimation,
none of existing works builds the inverse problem on O-flows.
\end{rem}

\vspace{0cm}

\begin{rem}
A model similar to (\ref{eq:O2Link-SingleStep}) can be constructed
based on the D-flow as well (D stands for destination) which also
allows for dimensional reduction and the preservation of OD flow information.
More specifically, let $x_{d}$ be the flows ending at the node $d$,
and $p_{ij,d}$ be the proportion of $x_{d}$ passing the link from
$i$ to $j$. Then OD flows can be computed via 
\[
s_{od}=x_{d}\left(\sum_{j}p_{oj,d}-\sum_{i}p_{io,d}\right)
\]
for single step model, and 
\begin{equation}
s_{od}^{t}=\sum_{\tau=1}^{\tau_{\max}}x_{d}^{t+\tau}\left(\sum_{j}p_{oj,d}^{\tau}-\sum_{i}p_{io,d}^{\tau}\right)\label{eq:DFlow2ODFlow-multiple}
\end{equation}
for multi-step model. The slight difference between (\ref{eq:DFlow2ODFlow-multiple})
and (\ref{eq:OFlow2ODFlow-multiple}) comes from the definition of
$s_{od}^{t}$ which describes the OD flow $s_{od}$ \emph{starting}
at the time interval $t$. The models based on O-flow and D-flow are
interchangeable. This paper focuses on the O-flow model only. 
\end{rem}

\vspace{0cm}

\subsection{\label{subsec:Lost-Information}Connections and Differences of Models}

On one hand, our model preserves the OD flow information via (\ref{eq:OFlow2ODFlow-single},\ref{eq:OFlow2ODFlow-multiple}).
On the other hand, the number of unknown variables (including those
in the assignment matrices and the flows) in our new model are substantially
reduced from the standard model (\ref{eq:OD2Link},\ref{eq:OD2Link-MultiStep}).
An educated instinct is that some information\footnote{This paper does not assume particular statistical models for the flows.
The term ``information'' here is not referred to as Shannon entropy
type of information.} must get lost by this dimension reduction. To characterize the exact
information got lost and to understand its importance, we need to
study the relationship among three different models built upon path
flows, OD flows, and O-flows, respectively. For simplicity of discussion,
we focus on a single snapshot of single step models.

Many papers \cite{BIERLAIRE2002TotalDemandScale,bierlaire2004efficient,bauer2017quasi}
in the literature use path flow models. For a given path $\bm{p}$
specified by $o\rightarrow i\rightarrow\cdots\rightarrow d$, denote
the flow along this particular path by $s_{\bm{p}}^{{\rm path}}$.
Group all path flows into a vector to form path flow vector $\bm{s}^{{\rm path}}$.
Then the observed link flows are given by 
\[
\bm{y}=\bm{A}^{{\rm path}}\bm{s}^{{\rm path}},
\]
where $\bm{A}^{{\rm path}}$ is the path incidence matrix where 
\[
A_{ij,\bm{p}}^{{\rm path}}=\begin{cases}
1 & \text{ if the path }\bm{p}\text{ involves the link }ij,\\
0 & \text{ otherwise.}
\end{cases}
\]

\begin{figure}

\begin{centering}
\includegraphics[scale=0.3]{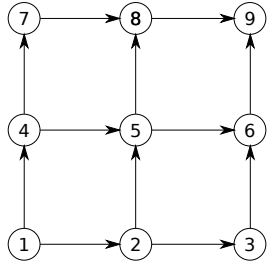} 
\par\end{centering}
\caption{\label{fig:3-by-3-network-unidirection}An example of 3-by-3 unidirectional
network. }

\end{figure}

The path flow model contains more information than the OD flow model
does. Let $\mathcal{P}_{od}$ be the set of all paths with origin
$o$ and destination $d$. It is straightforward to verify that 
\begin{align}
 & s_{od}=\sum_{\bm{p}\in\mathcal{P}_{od}}s_{\bm{p}}^{{\rm path}},\;{\rm and}\nonumber \\
 & a_{ij,od}=\frac{\sum_{\bm{p}\in\mathcal{P}_{od}}a_{ij,\bm{p}}^{{\rm path}}s_{\bm{p}}^{{\rm path}}}{\sum_{\bm{p}\in\mathcal{P}_{od}}s_{\bm{p}}^{{\rm path}}}.\label{eq:PathTrafficAssignment2ODTrafficAssignment}
\end{align}
However, the OD flow model does not fully characterize the path flow
model in general. Specifically, we use the following example to illustrate
that the lost information is the path selection information when multiple
paths associated with the same OD share the same subset of links. 
\begin{example}
\label{exa:InfLoss-Path2OD}In the unidirectional network in Figure
\ref{fig:3-by-3-network-unidirection}, consider the four different
paths $\bm{p}_{1}=1\rightarrow2\rightarrow5\rightarrow6\rightarrow9$,
$\bm{p}_{2}=1\rightarrow2\rightarrow5\rightarrow8\rightarrow9$, $\bm{p}_{3}=1\rightarrow4\rightarrow5\rightarrow6\rightarrow9$,
and $\bm{p}_{4}=1\rightarrow4\rightarrow5\rightarrow8\rightarrow9$
which correspond the same OD trip from 1 to 9. Suppose that the observed
link flows on all involved links are $c_{l}$ with $c_{l}>0$. The
corresponding OD flow model is unique: $s_{19}=2c_{l}$ and $a_{ij,19}=\frac{1}{2}$
for all involved links. However, the consistent path flow model is
not unique: that $s_{\bm{p}_{1}}^{{\rm path}}=s_{\bm{p}_{4}}^{{\rm path}}=c_{l}$
and $s_{\bm{p}_{2}}^{{\rm path}}=s_{\bm{p}_{3}}^{{\rm path}}=0$,
that $s_{\bm{p}_{1}}^{{\rm path}}=s_{\bm{p}_{4}}^{{\rm path}}=0$
and $s_{\bm{p}_{2}}^{{\rm path}}=s_{\bm{p}_{3}}^{{\rm path}}=c_{l}$,
and any convex combination of these two cases can give the exactly
same link flows. 
\end{example}

\vspace{0cm}

Similarly the OD flow model contains more information than the O-flow
model does. The O-flow model can be derived from the OD flow model
via 
\begin{align}
 & x_{o}=\sum_{d\ne o}s_{od},\;{\rm and}\nonumber \\
 & p_{ij,o}=\frac{\sum_{d\ne o}a_{ij,od}s_{od}}{\sum_{d\ne o}s_{od}},\label{eq:ODTrafficAssignment2OTrafficAssignment}
\end{align}
but not vice versa. An example is given below to show that the path
selection information gets lost in the shared part of the paths of
different OD flows. 
\begin{example}
\label{exa:InfLoss-OD2OFlow}Consider the same unidirectional network
in Figure \ref{fig:3-by-3-network-unidirection}. Consider the OD
flows $s_{18}$ and $s_{16}$. Further assume that the OD trip from
1 to 8 only has two possible paths $\bm{p}_{1}=1\rightarrow2\rightarrow5\rightarrow8$
and $\bm{p}_{2}=1\rightarrow4\rightarrow5\rightarrow8$, and that
the OD trip from 1 to 6 only has two possible paths $\bm{p}_{3}=1\rightarrow2\rightarrow5\rightarrow6$
and $\bm{p}_{4}=1\rightarrow4\rightarrow5\rightarrow6$. Suppose that
the observed link flows of all involved links are $c_{l}$ with $c_{l}>0$.
This uniquely identifies the O-flow model: $x_{1}=2c_{l}$ and $p_{ij,1}=\frac{1}{2}$
for all involved links. However, the consistent OD flow model is not
unique: while the OD flows are uniquely given by $s_{18}=s_{16}=c_{l}$,
the choices of the assignment matrix are not unique; that $s_{18}=s_{\bm{p}_{1}}^{{\rm path}}$
($a_{12,18}=1$ and $a_{14,18}=0$) and $s_{16}=s_{\bm{p}_{4}}^{{\rm path}}$
($a_{12,16}=0$ and $a_{14,16}=1$), that $s_{18}=s_{\bm{p}_{2}}^{{\rm path}}$
($a_{12,18}=0$ and $a_{14,18}=1$) and $s_{16}=s_{\bm{p}_{3}}^{{\rm path}}$
($a_{12,16}=1$ and $a_{14,16}=0$), and any convex combination of
these two cases give the exactly same link flows. 
\end{example}

\section{\label{sec:Blind-Estimation} Estimation of OD flows}

The substantial reduction of dimension in our new forward model has
the potential to turn the ill-posedness of the OD estimation problem
to well-posedness. This section is devoted to a joint estimation of
both O-flows and the assignment matrix, which leads to an estimation
of OD flows based on Theorem \ref{thm:OD-info-preserved}. 

\subsection{Joint Estimation }

Let $\bm{y}^{t}$, $t=1,2,\cdots,n_{T}$, be the measured link flows,
where $n_{T}$ denotes the time horizon. Assume that $\tau_{\max}<n_{T}$.
Then the joint estimation problem can be formulated as 
\[
\underset{\bm{P},\;\bm{x}^{t}:\;1\le t\le n_{T}}{\min}\;f_{{\rm cost}}\left(\bm{P},\left\{ \bm{x}^{t}\right\} \right)
\]
for the single step model and 
\[
\underset{\begin{array}{c}
\bm{P}^{t}:\;1\le t\le\tau_{\max}\\
\bm{x}^{t}:\;2-\tau_{\max}\le t\le n_{T}
\end{array}}{\min}\;f_{{\rm cost}}\left(\left\{ \bm{P}^{t}\right\} ,\left\{ \bm{x}^{t}\right\} \right)
\]
for the multi-step model. The $f_{{\rm cost}}$ relies on the underlying
assumption on noise and typically is chosen to be a convex function
for computational convenience. In this paper, we adopt the most commonly
used one, the squared $\ell_{2}$-norm, for $f_{{\rm cost}}$, resulting
in

\begin{equation}
\underset{\begin{array}{c}
\bm{P},\;\left\{ \bm{x}^{t}\right\} \end{array}}{\min}\;\sum_{t=1}^{n_{T}}\left\Vert \bm{y}^{t}-\bm{P}\bm{x}^{t}\right\Vert _{2}^{2}\label{eq:Joint-Estimation_single}
\end{equation}
for the single step model, and 
\begin{equation}
\underset{\begin{array}{c}
\left\{ \bm{P}^{t}\right\} ,\;\left\{ \bm{x}^{t}\right\} \end{array}}{\min}\;\sum_{t=1}^{n_{T}}\left\Vert \bm{y}^{t}-\sum_{\tau=1}^{\tau_{\max}}\bm{P}^{\tau}\bm{x}^{t-\tau+1}\right\Vert _{2}^{2}\label{eq:Joint-Estimation_multiple}
\end{equation}
for the multi-step model.

The following constraints imposed by the physics and network topology
should be also considered. 
\begin{enumerate}
\item[C1:] Non-negativity of flows. That is, 
\[
\bm{x}^{t}\ge\bm{0},\;\forall t.
\]
\item[C2:] Probability constraint on assignment matrix. By the definition of
traffic assignment, 
\[
0\leq p_{ij,o}^{t}\leq1
\]
for multi-step models, and 
\[
0\leq p_{ij,o}\leq1
\]
for single step models. 
\item[C3:] Observability constraint. This means that all the traffic originated
from a node is observed. That is, for multi-step models one has 
\[
\sum_{\tau=1}^{\tau_{\max}}\sum_{i\ne o}p_{oi,o}^{\tau}=1,\;\forall o,
\]
and for single step models it holds that 
\[
\sum_{i\ne o}p_{oi,o}=1,\;\forall o.
\]
This constraint is consistent with the principles of flow conservation
mentioned in \cite{lo1996estimation}. 
\item[C4:] \label{enu:speed-constraint}Speed constraint.

For multi-step models, one has that $p_{ij,o}^{t}=0$ if the link
$ij$ is not involved in the $t$-th step of the O-flow $x_{o}$.
For example, under the rigid multi-step model (see Remark \ref{rem:rigid-model}
for more details), $p_{oi,o}^{t}\ge0$ for $t=1$, $p_{oi,o}^{t}=0\;{\rm for}\;t>1$,
and $p_{ij,o}^{t}=0\;{\rm for}\;t=1,\;\forall o,\;\forall i\ne o$.
The complete set of speed constraint depends on the network topology
and the speed of the traffic. Application of the speed constraint
typically makes the assignment matrices sparse especially when the
sampling time interval is small compared to maximum trip time.

For single step models, it holds that $p_{ij,o}=0$ when a link $ij$
is not involved in any trips from $o$. The assignment matrix is sparse
when the maximum trip time is much smaller than the diameter of the
graph. 
\item[C5:] \label{enu:Flow-constraint}Flow constraint. Consider the flow originated
from any node $o$. For any given node $i\ne o$, its inflow must
be larger than or equal to its outflow. For multi-step models, one
has 
\[
\sum_{\tau=1}^{\tau_{\max}}\sum_{j}p_{ji,o}^{\tau}-\sum_{\tau=t}^{\tau_{\max}}\sum_{k}p_{ik,o}^{\tau}\geq0,\;\forall t,\;\forall o,\;{\rm and}\;\forall i\ne o,
\]
which can be simplified into 
\[
\sum_{j}p_{ji,o}^{t-1}-\sum_{k}p_{ik,o}^{t}\geq0,\;\forall t,\;\forall o,\;{\rm and}\;\forall i\ne o.
\]
for rigid models (Remark \ref{rem:rigid-model}). For single step
models, it holds that 
\[
\sum_{j}p_{ji,o}-\sum_{k}p_{ik,o}\geq0,\;\forall o,\;{\rm and}\;\forall i\ne o.
\]
\end{enumerate}
\begin{rem}
\label{rem:rigid-model}At this point, it is worth to distinguish
two traffic models: rigid and elastic traffic models respectively.
In the rigid multi-step model, all links are of the same length, and
in one unit time all the unfinished traffic flows move forward across
exactly one link. The second assumption is equivalent to say that
all the traffic flows have constant and identical speeds. In the elastic
model, the constraints on both link length and traffic speed are removed.
The elastic model fits actual systems better. In the simulation part
of this paper (Section \ref{sec:Simulations}), we adopt the rigid
model for the simplicity of modeling and simulations. 
\end{rem}

\vspace{0.01cm}

The problem (\ref{eq:Joint-Estimation_multiple}) is known as a bilinear
inverse problem \cite{davenport2016overview}. The name comes from
the fact that when one fixes either $\left\{ \bm{P}^{t}\right\} $
or $\left\{ \bm{x}^{t}\right\} $ and solves for the other, the inverse
problem becomes a linear inverse problem. Generally speaking, bilinear
inverse problems are non-convex and does not admit a unique solution. 
\begin{rem}[Extensions to Accommodate Prior Information]
The formulation (\ref{eq:Joint-Estimation_multiple}) can be extended
to accommodate prior information by adding extra terms into the objective
function. For example, when approximations of $\left\{ \bm{P}^{t}\right\} $
and $\left\{ \bm{x}^{t}\right\} $, denoted by $\left\{ \tilde{\bm{P}}^{t}\right\} $
and $\left\{ \tilde{\bm{x}}^{t}\right\} $, are available from historical
data \cite{menon2015fine,toledo2013estimation,parry2012estimation},
one may formulate the estimation problem via a form similar to the
GLS formulation (\ref{eq:GLS OD estimator old}): 
\begin{align*}
\min_{\left\{ \bm{P}^{t}\right\} ,\;\left\{ \bm{x}^{t}\right\} }\; & \sum_{t=1}^{n_{T}}\left\Vert \bm{y}^{t}-\sum_{\tau=1}^{\tau_{\max}}\bm{P}^{\tau}\bm{x}^{t-\tau+1}\right\Vert _{2}^{2}\\
 & +\sum_{t=1}^{\tau_{\max}}{\rm vect}\left(\bm{P}^{t}-\tilde{\bm{P}}^{t}\right)^{T}\bm{W}_{P}^{-1}{\rm vect}\left(\bm{P}^{t}-\tilde{\bm{P}}^{t}\right)+\sum_{t=2-\tau_{\max}}^{n_{T}}\left(\bm{x}^{t}-\tilde{\bm{x}}^{t}\right)^{T}\bm{W}_{X}^{-1}\left(\bm{x}^{t}-\tilde{\bm{x}}^{t}\right),
\end{align*}
where ${\rm vect}\left(\cdot\right)$ denotes the vector formed by
stacking the columns of the input matrix, and $\bm{W}_{P}$ and $\bm{W}_{X}$
are the covariance matrices of the error terms.
\end{rem}

We present the Gauss-Seidel method for OD estimation based on multi-step
models. It is an alternative minimization approach where in each iteration
either $\left\{ \bm{P}^{t}\right\} $ or $\left\{ \bm{x}^{t}\right\} $
is fixed and the minimization is with respect to the other variable.
See Algorithm \ref{alg:AlternativeOpt} for a high level description. 

\begin{algorithm}
\caption{Alternative minimization for Joint Estimation\label{alg:AlternativeOpt}}

\begin{algorithmic}[1]

\renewcommand{\algorithmicrequire}{\textbf{Input:}} 
\renewcommand{\algorithmicensure}{\textbf{Output:}}
\Require the graph $\mathcal{G}\left(\mathcal{N},\mathcal{L}\right)$ and link flows $\bm{y}^t$, $t=1,2,\cdots,n_{T}$
\Ensure $\bm{P}^{t}$, $t=1,2,\cdots,\tau_{\max}$, and $\bm{x}^{t}$, $t=1,2,\cdots,n_{T}$
\renewcommand{\algorithmicensure}{\textbf{Initialization:}}
\Ensure Based on $\mathcal{G}\left(\mathcal{N},\mathcal{L}\right)$ randomly initialize $\left\{ \bm{P}^{t}\right\}$ subject to Constraints C2-C5.
\While{stop criteria are not satisfied,}
\State With fixed $\left\{ \bm{P}^{t}\right\}$, update $\left\{ \bm{x}^{t}\right\}$ to minimize $f_{{\rm cost}}$ subject to Constraint C1.
\State With fixed $\left\{ \bm{x}^{t}\right\}$, update $\left\{ \bm{P}^{t}\right\}$ to minimize $f_{{\rm cost}}$ subject to Constraints C2-C5.
\EndWhile
\end{algorithmic}
\end{algorithm}

\begin{rem}
\label{rem:time-horizon-Oflows}A subtle point is the time intervals
of the output $\bm{x}^{t}$. From the definition of the convolution,
$\bm{x}^{t},\;2-\tau_{\max}\le t\le n_{T}$, is involved in producing
$\bm{y}^{t},\;1\le t\le n_{T}$. However, it is clear that with given
information $\bm{y}^{t},\;1\le t\le n_{T}$, it is impossible to uniquely
recover $\bm{x}^{t}$ for $2-\tau_{\max}\le t\le0$ in general. As
a consequence, the outputs only involve $\bm{x}^{t}$ for $1\le t\le n_{T}$. 
\end{rem}

\subsection{\label{subsec:Uniqueness-Discussion}Uniqueness of Solution}

In the discussion of the uniqueness, we assume the noise free case
where $\bm{y}^{t}=\bm{P}^{t}*\bm{x}^{t}$. Denote the estimated assignment
matrices and O-flows by $\left\{ \hat{\bm{P}}^{t}\right\} $ and $\left\{ \hat{\bm{x}}^{t}\right\} $,
respectively. Feasible solutions satisfy $\bm{y}^{t}=\hat{\bm{P}}^{t}*\hat{\bm{x}}^{t}$.
The uniqueness of the solution implies that the estimated $\left\{ \hat{\bm{P}}^{t}\right\} $
and $\left\{ \hat{\bm{x}}^{t}\right\} $ are the same as the ground
truth.

The solution of a general bilinear inverse problem is not unique.
If $\left\{ \hat{\bm{P}}^{t}\right\} $ and $\left\{ \hat{\bm{x}}^{t}\right\} $
are a solution of $\bm{y}^{t}=\bm{P}^{t}*\bm{x}^{t}$, $1\le t\le n_{T}$,
then so are $\left\{ \hat{\bm{P}}^{t}\bm{M}\right\} $ and $\left\{ \bm{M}^{-1}\hat{\bm{x}}^{t}\right\} $
for arbitrary invertible matrix $\bm{M}$. In our problem, the equality
constrains imposed by Constraints C3 and C4 help avoid such ambiguity. 

\subsubsection{\label{subsec:Uniqueness-of-Bidirectional-network}A Necessary Condition
for the Uniqueness of Solution }

The following theorem states a necessary condition for the uniqueness
of solution of Equations (\ref{eq:O2Link-SingleStep}) and (\ref{eq:O2Link-MultiStep}).
Let $n_{O}$ be the number of valid origin nodes (in bidirectional
networks $n_{O}=n_{n}$), which gives the number of equations obtained
from the observability constraint C3. Let $n_{4}$ denote the number
of equations obtained from the speed constraint C4. The values of
$n_{O}$ and $n_{4}$ depend on the network topology. 
\begin{thm}
\label{thm:Uniqueness}If the solution of the single step model (\ref{eq:O2Link-SingleStep})
is unique, then it holds that 
\begin{equation}
n_{\ell}n_{T}+n_{O}+n_{4}\ge n_{\ell}n_{O}+n_{O}n_{T},\label{eq:necessary-cond-single-1}
\end{equation}
or equivalently 
\begin{equation}
n_{T}\ge\frac{\left(n_{\ell}-1\right)n_{O}-n_{4}}{n_{\ell}-n_{O}}.\label{eq:necessary-cond-single-2}
\end{equation}
That the solution of the multi-step model (\ref{eq:O2Link-MultiStep})
is unique implies that 
\begin{equation}
n_{\ell}n_{T}+n_{O}+n_{4}\ge\tau_{\max}n_{\ell}n_{O}+n_{O}n_{T},\label{eq:necessary-cond-multistep-1}
\end{equation}
or equivalently 
\begin{equation}
n_{T}\ge\frac{\left(\tau_{\max}n_{\ell}-1\right)n_{O}-n_{4}}{n_{\ell}-n_{O}}.\label{eq:necessary-cond-multistep-2}
\end{equation}
\end{thm}

\begin{proof}
The proof is based on the algebraic argument that the solutions cannot
be unique if the number of independent equations is less than the
number of unknown variables. The left hand sides of Inequalities  (\ref{eq:necessary-cond-single-1})
and (\ref{eq:necessary-cond-multistep-1}) are the total number of
equations, and the right hand sides of Inequalities (\ref{eq:necessary-cond-single-1})
and (\ref{eq:necessary-cond-multistep-1}) are the number of unknowns
for single step models and multi-step models, respectively. 
\end{proof}
\vspace{0cm}

To estimate an $n_{T}$ necessary for the uniqueness of solution,
a relaxed and hence less accurate version can be derived as follows.
Let $c=n_{\ell}/n_{O}$ be the average number of links per node. We
relax the right hand sides of Inequalities (\ref{eq:necessary-cond-single-2})
and (\ref{eq:necessary-cond-multistep-2}) by ignoring the terms $-n_{O}$
and $-n_{4}$ in the numerators. Then the useful rules of thumb for
practice can be obtained where 
\begin{equation}
n_{T}\gtrapprox\frac{c}{c-1}n_{O}\label{eq:necessray-cond-single-simple}
\end{equation}
for single step models, and 
\begin{equation}
n_{T}\gtrapprox\frac{\tau_{\max}c}{c-1}n_{O}\label{eq:necessary-cond-multistep-simple}
\end{equation}
for multi-step models.

\subsubsection{\label{subsec:Uniqueness-Unidirectional-Networks}Non-uniqueness
of Solutions for Unidirectional Networks and A Remedy}

In the following, we shall use the necessary condition to show that
in general unidirectional networks do not admit a unique solution.
We start the analysis by studying a three node unidirectional network.
For simplicity, we focus on the rigid model mentioned in Remark \ref{rem:rigid-model}
when considering multi-step models. 

\begin{wrapfigure}{o}{0.5\columnwidth}%

\begin{centering}
\includegraphics[scale=0.3]{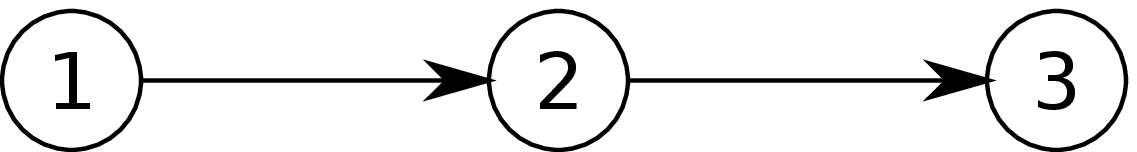} 
\par\end{centering}
\caption{\label{fig:3-node-unidirectional-network}A three node unidirectional
network}

\end{wrapfigure}%

\begin{lem}
\label{lem:nonuniqueness-3-node-unidirectional-network}The three
node unidirectional network in Figure \ref{fig:3-node-unidirectional-network}
does not admit a unique solution in general. (The case that admits
a unique solution is detailed in the proof.) 
\end{lem}

\begin{proof}
We start with the single step model. For the unidirectional network
in Figure \ref{fig:3-node-unidirectional-network}, Constraints C3
and C4 imply that $p_{12,1}=1$, $p_{12,2}=0$, and $p_{23,2}=1$.
The system becomes 
\[
\left[\cdots\begin{array}{c}
y_{12}^{t}\\
y_{23}^{t}
\end{array}\cdots\right]=\left[\begin{array}{cc}
1 & 0\\
p_{23,1} & 1
\end{array}\right]\left[\cdots\begin{array}{c}
x_{1}^{t}\\
x_{2}^{t}
\end{array}\cdots\right],\;1\le t\le n_{T}.
\]
The number of equations is $2n_{T}$ while the number of unknowns
is $1+2n_{T}$. Therefore, this network does not admit a unique solution
in general.

A more careful study reveals that $x_{1}^{t}=y_{12}^{t}$, $\forall t$,
and 
\begin{equation}
p_{23,1}\in\left[0,\min\left(1,\underset{1\le t\le n_{T}}{\min}\frac{y_{23}^{t}}{x_{1}^{t}}\right)\right]=\left[0,\min\left(1,\underset{1\le t\le n_{T}}{\min}\frac{y_{23}^{t}}{y_{12}^{t}}\right)\right].\label{eq:p2-3-t}
\end{equation}
A unique solution exists if and only if $y_{23}^{t}=0$ and $y_{12}^{t}\ne0$
for some $t\in\left\{ 1,\cdots,n_{T}\right\} $, which corresponds
the solution that $p_{23,1}=0$ and $y_{23}^{t}=x_{2}^{t}$.

The proof for the multi-step model is similar and hence omitted here.
The solution is not unique except that $y_{23}^{t+1}=0$ and $y_{12}^{t}\ne0$
for some $t\in\left\{ 1,\cdots,n_{T}-1\right\} $, in which case $p_{23,1}^{2}=0$. 
\end{proof}
\vspace{0cm}

The following corollary is a direct application of Lemma \ref{lem:nonuniqueness-3-node-unidirectional-network}. 
\begin{cor}
\label{cor:nonuniqueness-unidirectional-network}Unidirectional networks
that contain Figure \ref{fig:3-node-unidirectional-network} as a
subgraph do not admit a unique solution in general. 
\end{cor}

\vspace{0cm}

The above negative result for unidirectional networks can be addressed
by assuming that the O-flows are sparse. In particular, define the
vector $\bm{x}_{o}=\left[x_{o}^{1},\cdots,x_{o}^{n_{T}}\right]^{T}$,
$\forall o\in\left\{ 1,\cdots,n_{n}\right\} $. This vector is called
\emph{sparse} if the number of non-zeros entries in $\bm{x}_{0}$
is much less than the zero entries.

We show how sparsity assumption results in a unique solution for the
three-node unidirectional network in Figure \ref{fig:3-node-unidirectional-network}.
Recall the range of $p_{23,1}$ given in (\ref{eq:p2-3-t}) in the
proof for Lemma \ref{lem:nonuniqueness-3-node-unidirectional-network}.
Suppose that 
\[
\underset{1\le t\le n_{T}}{\min}\frac{y_{23}^{t}}{y_{12}^{t}}\le1.
\]
Then the solution 
\[
p_{23,1}=\underset{1\le t\le n_{T}}{\min}\frac{y_{23}^{t}}{y_{12}^{t}}
\]
gives the sparsest solution $\bm{x}^{t}$ where $x_{2}^{t}=0$ for
all 
\[
t\in\left\{ \tilde{t}:\;\frac{y_{23}^{\tilde{t}}}{y_{12}^{\tilde{t}}}=\underset{1\le t^{\prime}\le n_{T}}{\min}\frac{y_{23}^{t^{\prime}}}{y_{12}^{t^{\prime}}}\right\} .
\]
All other solutions 
\[
p_{23,1}\in\left[0,\underset{1\le t\le n_{T}}{\min}\frac{y_{23}^{t}}{x_{1}^{t}}\right)
\]
result in less sparse solutions, i.e., $x_{2}^{t}>0$ for all $t$.

For practical usage, we assume that O-flows are sparse in a transform
domain. For a given origin $o$, represent the time series of the
corresponding O-flows in a vector form $\bm{x}_{o}=\left[x_{o}^{1},\cdots,x_{o}^{n_{T}}\right]^{T}$.
We assume that $\bm{x}_{o}$ is sparse under an invertible transform
$\bm{D}$, i.e., most entries of $\bm{c}_{o}:=\bm{D}\bm{x}_{o}$ are
zeros and $\bm{D}^{-1}$ exists. The sparsity in the transform domain
implies that the O-flows exhibit some temporal patterns, which matches
everyday experience. In the simulations of this paper, the transform
is chosen to be discrete cosine transform (DCT) where the transformation
matrix $\bm{D}$ is orthonormal, i.e., $\bm{D}^{-1}=\bm{D}^{T}$. 

This transform domain sparsity promotes a unique solution. For simplicity
of discussion, consider the single step model which can be written
as $\left[\bm{y}^{1},\cdots,\bm{y}^{n_{T}}\right]=\bm{P}\left[\bm{x}^{1},\cdots,\bm{x}^{n_{T}}\right]=\bm{P}\left[\cdots,\bm{x}_{o},\cdots\right]^{T}=\bm{P}\bm{C}^{T}\bm{D}^{-T}$
where $\bm{C}=\left[\cdots,\bm{c}_{o},\cdots\right]$. Define $\tilde{\bm{Y}}=\bm{Y}\bm{D}^{T}$.
Then $\tilde{\bm{Y}}=\bm{P}\bm{C}^{T}$ where each row of the matrix
$\bm{C}^{T}$ is sparse. Let $\mathcal{T}_{o,{\rm nz}}=\left\{ t:\;c_{o}^{t}\ne0\right\} $
for any given origin $o$. We conjure that if $\mathcal{T}_{o,{\rm nz}}\nsubseteq\bigcup_{o^{\prime}}\mathcal{T}_{o^{\prime},{\rm nz}}$
for all $o\in\left\{ 1,\cdots,n_{n}\right\} $, then the equivalent
single step model $\tilde{\bm{Y}}=\bm{P}\bm{C}^{T}$ admits a unique
solution with sufficiently large $n_{T}$ and diverse $\bm{C}^{T}$.

The sparsity assumption used in this paper is significantly different
from that in \cite{menon2015fine} (see Section \ref{subsec:OD-Estimation-Literature}
for more detailed discussions on \cite{menon2015fine}). In \cite{menon2015fine},
it assumes that among all the OD flows, only a small fraction of them
are significant and the rest can be safely approximated zero. The
way to identify significant OD flows is to use historical data. By
contrast, in our assumption all O-flows are nonzero and only the transform
coefficients are sparse. Our assumption has wider applicability for
several reasons. First, the assumption involving zero flows may not
be true when analyzing busy communities. Second, in applications with
no or little historical data, there is not enough prior knowledge
to identify which O-flows are more significant than the others. 

\subsection{\label{subsec:Algorithm-sparsity} Estimation with the Sparsity Assumption}

As shown in Section \ref{subsec:Uniqueness-Unidirectional-Networks},
the transform domain sparsity promotes a unique solution. Following
the notations at the end of Section \ref{subsec:Uniqueness-Unidirectional-Networks},
we assume that $\bm{x}_{o}$ is sparse under an orthonormal transformation
matrix $\bm{D}$, that is, $\bm{c}_{o}=\bm{D}\bm{x}_{o}$ and the
number of nonzeros in $\bm{c}_{o}$ is much less than $n_{T}$. To
enforce this sparsity constraint, the estimation problem can be written
as
\begin{align*}
\min_{\{\bm{P}^{t}\},\;\{\bm{x}^{t}\}} & \;\sum_{t=1}^{n_{T}}\frac{1}{2}\left\Vert \bm{y}^{t}-\sum_{\tau=1}^{\tau_{\max}}\bm{P}^{\tau}\bm{x}^{t-\tau+1}\right\Vert _{2}^{2}+\sum_{o}\lambda_{o}\left\Vert \bm{D}\bm{x}_{o}\right\Vert _{0},
\end{align*}
where $\lambda_{o}\ge0$, $\forall o$, are appropriately chosen parameters,
and $\left\Vert \cdot\right\Vert _{0}$ denotes the $\ell_{0}$ pseudo-norm
which counts the number of nonzero elements. The $\ell_{0}$ pseudo-norm
is not convex. The common practice is to replace it with its convex
envelope $\ell_{1}$-norm where $\left\Vert \bm{c}\right\Vert _{1}:=\sum_{j}\left|c_{j}\right|$.
One has 
\begin{equation}
\min_{\{\bm{P}^{t}\},\;\{\bm{x}^{t}\}}\;\sum_{t}\frac{1}{2}\left\Vert \bm{y}^{t}-\sum_{\tau=1}^{\tau_{\max}}\bm{P}^{\tau}\bm{x}^{t-\tau+1}\right\Vert _{2}^{2}+\sum_{o}\lambda_{o}\left\Vert \bm{D}\bm{x}_{o}\right\Vert _{1}.\label{eq:Joint-Estimation-l1}
\end{equation}
This formulation looks similar to the famous LASSO \cite{tibshirani1996regression,donoho2006most}
except that the problem in (\ref{eq:Joint-Estimation-l1}) is bilinear
and hence non-convex.

The regularization parameters $\lambda_{o}$ in (\ref{eq:Joint-Estimation-l1})
have to be carefully chosen. According to the authors' knowledge,
there is no recipe to choose the optimal values of the regularization
parameters $\lambda_{o}$. Typically, their values are set by trial
and error.

To minimize possible efforts of parameter tuning, we assume that the
level of noise in link flow observations is known a priori (which
is acceptable for real applications) and design the following constrained
optimization approach. Suppose that the relative level of noise is
upper bounded by $\epsilon>0$, i.e., 
\[
\frac{\sum_{t}\left\Vert \bm{y}^{t}-\bm{P}^{t}*\bm{x}^{t}\right\Vert _{2}^{2}}{\sum_{t}\left\Vert \bm{y}^{t}\right\Vert _{2}^{2}}\leq\epsilon.
\]
We have the following formulation: 
\begin{align}
\min_{\{\bm{P}^{t}\},\;\{\bm{x}^{t}\}}\; & \sum_{o}\left\Vert \bm{D}\bm{x}_{o}\right\Vert _{1}\nonumber \\
{\rm subject\;to}\; & \sum_{t=1}^{n_{T}}\left\Vert \bm{y}^{t}-\bm{P}^{t}*\bm{x}^{t}\right\Vert _{2}^{2}\leq\epsilon\sum_{t=1}^{n_{T}}\left\Vert \bm{y}^{t}\right\Vert _{2}^{2}.\label{eq:L1-x-minimization-constrained}
\end{align}
The Gauss-Seidel method can be applied to solve the optimization problem
iteratively.

However, there is a technical issue for early iterations when the
Gauss-Seidel method is applied. In the early stage of the iterations,
the estimated $\left\{ \hat{\bm{P}}^{t}\right\} $ and $\left\{ \hat{\bm{x}}^{t}\right\} $
may not satisfy the constraint specified in (\ref{eq:L1-x-minimization-constrained}).
To address this issue, one may first run the Gauss-Seidel method for
the objective function (\ref{eq:Joint-Estimation_multiple}) until
the estimated $\left\{ \hat{\bm{P}}^{t}\right\} $ and $\left\{ \hat{\bm{x}}^{t}\right\} $
satisfy the constraint in (\ref{eq:L1-x-minimization-constrained}),
and then apply the Gauss-Seidel method for the problem (\ref{eq:L1-x-minimization-constrained}).
In the case that the convergence rate is too slow, one may first relax
$\epsilon>0$ to a larger value $\epsilon^{\prime}>\epsilon$, and
then swap back to $\epsilon$ after the algorithm converges for $\epsilon^{\prime}$.

\section{\label{sec:Simulations}Simulations and Results}

The purpose of the simulations is to demonstrate the possibility of
blind estimation. For this purpose, we generate our own synthetic
data and do not assume measurement noise. Furthermore, for implementation
simplicity, we focus on the rigid multi-step models specified in Remark
\ref{rem:rigid-model}. 

\subsection{\label{subsec:Generation-of-Networks}Simulation Setting}

Four networks in Figure \ref{fig:bidirectional-networks} are involved
in the numerical test. For each network, we consider all and only
the OD trips with distance at most 4 links. Based on this assumption,
basic information about these four networks is listed below. 
\begin{itemize}
\item The 3-by-3 unidirectional network in Figure \ref{fig:3-by-3-Uni-network}:
12 links, 27 OD pairs, and 8 origin nodes. 
\item The 3-by-3 bidirectional network in Figure \ref{fig:3-by-3-Bi-network}:
24 links, 72 OD pairs, and 9 origin nodes. 
\item The 8-by-8 bidirectional network in Figure \ref{fig:8-by-8-network}:
224 links, 1660 OD pairs, and 64 origin nodes. 
\item The G\'EANT network \cite{uhlig2006providing} in Figure \ref{fig:GEANT-network}:
74 links, 490 OD pairs, and 23 origin nodes. 
\end{itemize}
In the third network, the number of link flows is less than $\nicefrac{1}{7}$
of that of OD flows. The G\'EANT network listed the last is a real
network \cite{uhlig2006providing}, where the number of link flows
is slightly more than $\nicefrac{1}{7}$ of that of OD flows. For
both networks, conventional models based OD flows result in severely
ill-posed inverse problems.

\begin{figure}
\begin{centering}
\subfloat[\label{fig:3-by-3-Uni-network}]{\begin{centering}
\includegraphics[scale=0.3]{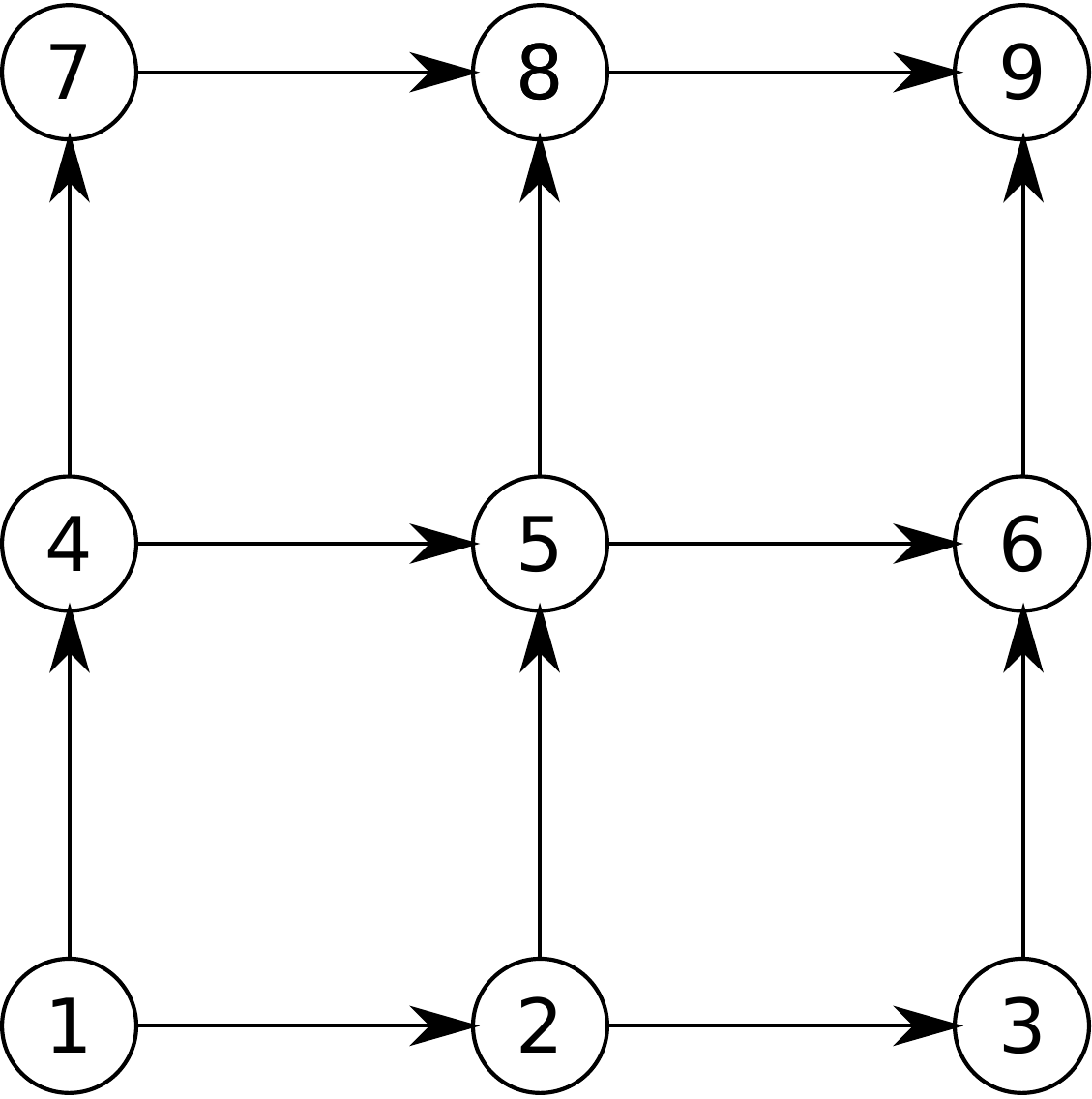} 
\par\end{centering}
}\hspace{0.5cm}\subfloat[\label{fig:3-by-3-Bi-network}]{\begin{centering}
\includegraphics[scale=0.3]{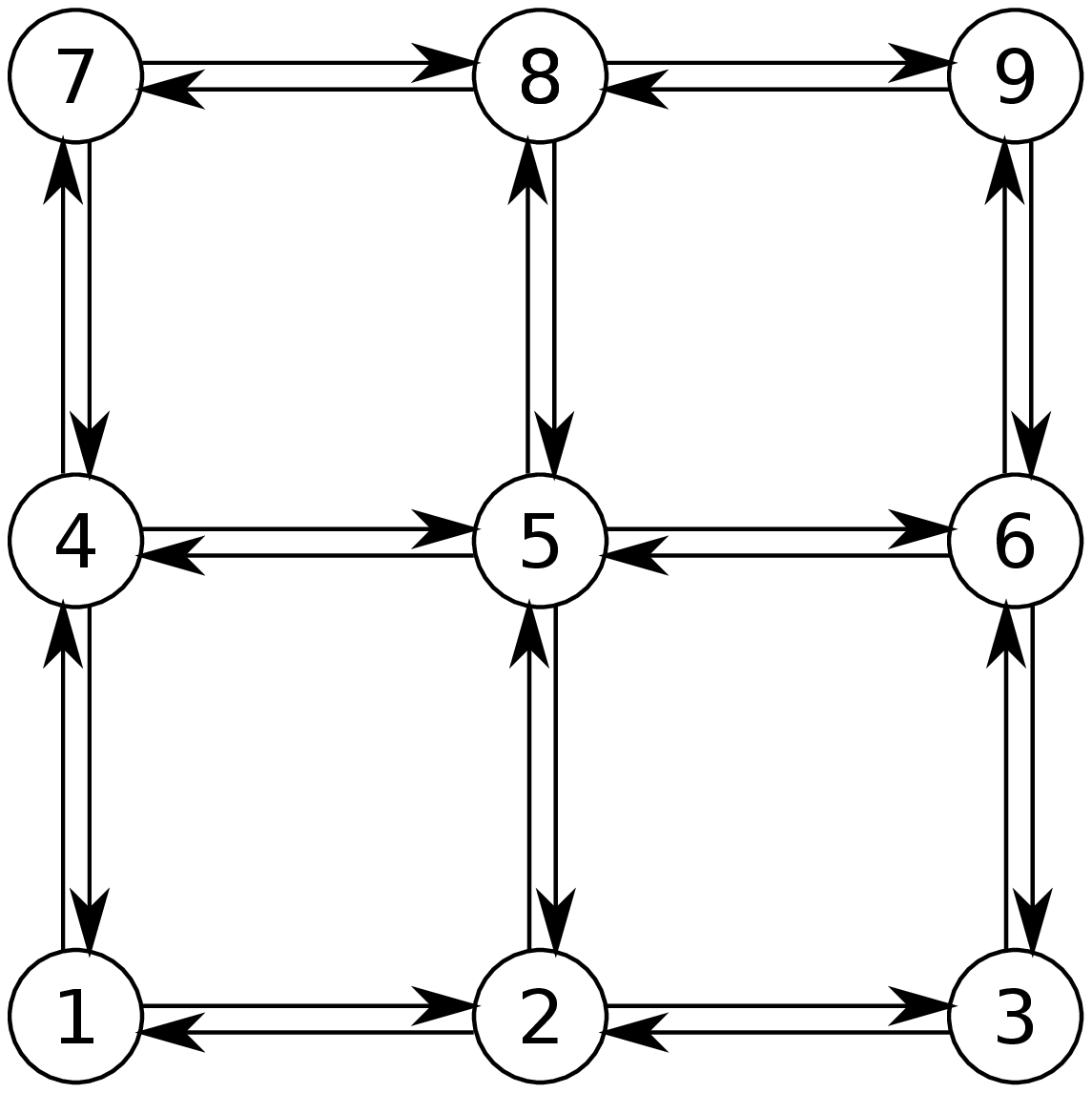} 
\par\end{centering}
}\hspace{0.5cm}\subfloat[\label{fig:8-by-8-network}]{\begin{centering}
\includegraphics[scale=0.3]{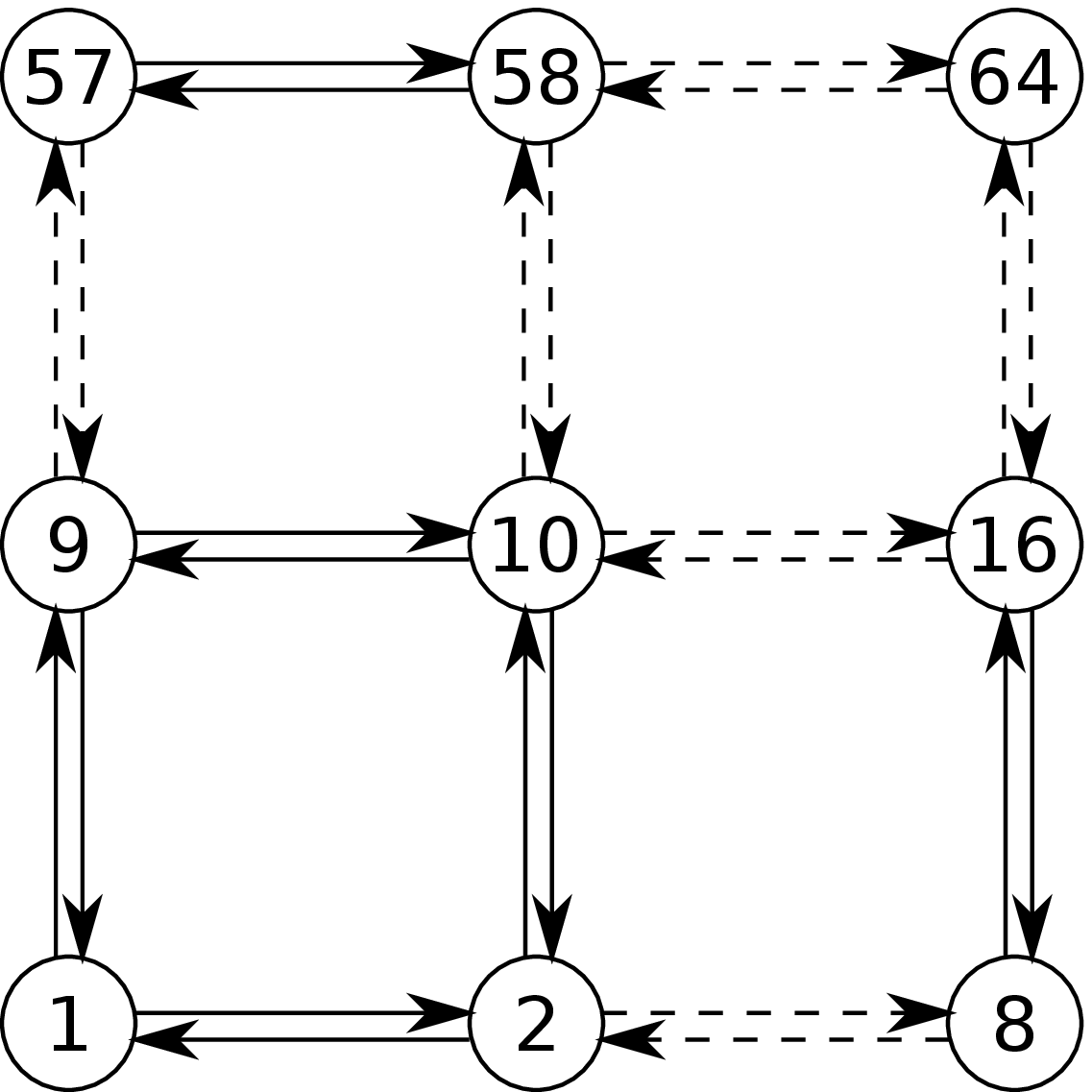} 
\par\end{centering}
}
\par\end{centering}
\begin{centering}
\subfloat[\label{fig:GEANT-network}]{\begin{centering}
\includegraphics[scale=0.6]{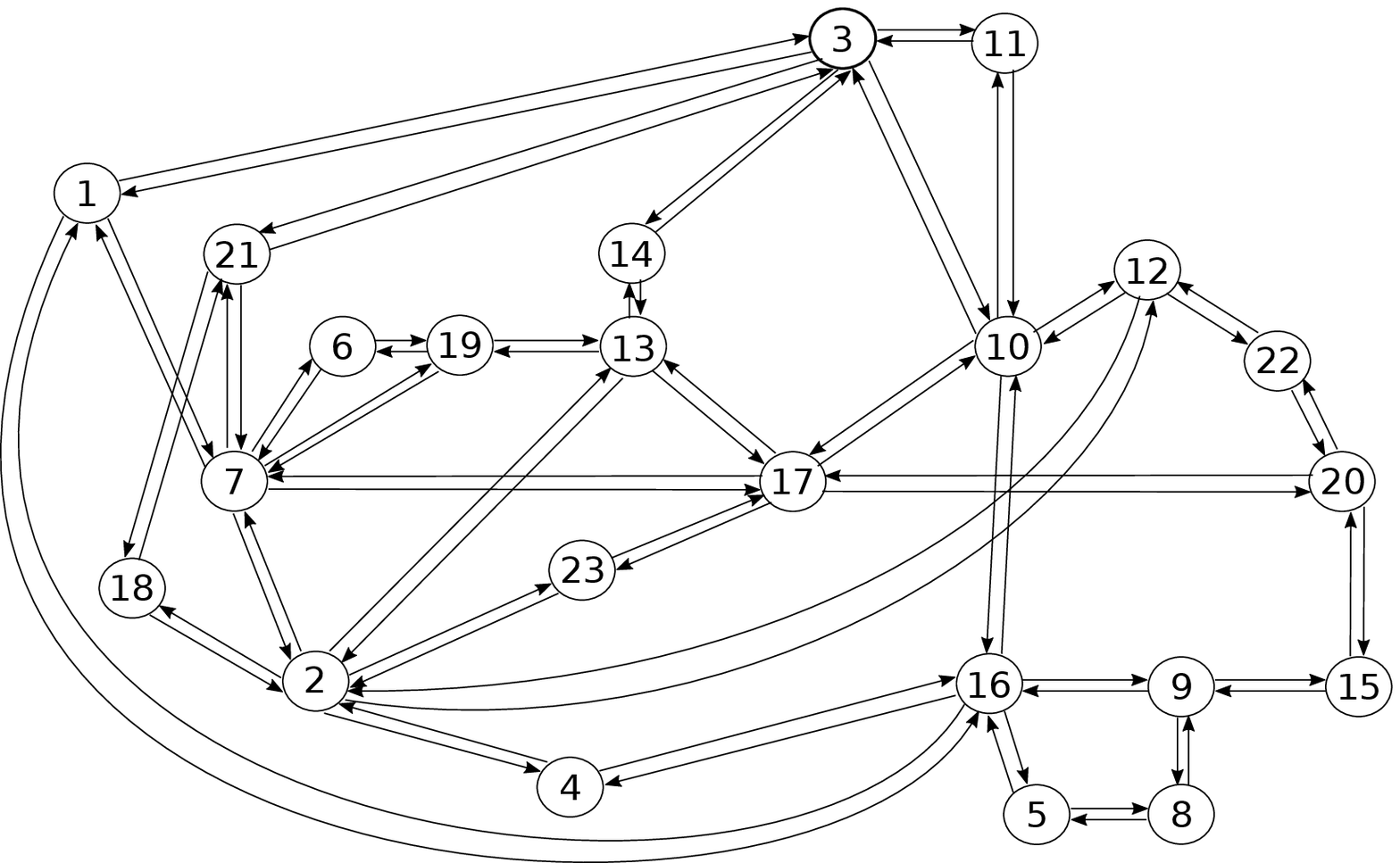} 
\par\end{centering}
}
\par\end{centering}
\caption{\label{fig:bidirectional-networks}Bidirectional networks used in
the simulations. A: A 3-by-3 Unidirectional network. B: A 3-by-3 bidirectional
network. C: An 8-by-8 bidirectional network. D: The G\'EANT network
\cite{uhlig2006providing}. }
\end{figure}

The ground truth data are generated as follows. 
\begin{itemize}
\item Firstly, we find out all valid paths and OD trips for each network.
For a given network $\mathcal{G}\left(\mathcal{N},\mathcal{L}\right)$,
we find all the loop-free paths that involve at most $\tau_{\max}$
many links and denote the corresponding list by $\mathcal{L}_{{\rm path}}$.
In this paper, we set $\tau_{\max}=4$. Using these paths, a list
of valid OD trips is generated and denoted by $\mathcal{L}_{{\rm OD}}$.
Based on $\mathcal{L}_{{\rm OD}}$, a list of origin nodes is constructed
and denoted by $\mathcal{L}_{{\rm O}}$. 
\item Then, we randomly generate the assignment matrix for the O-flow model. 
\begin{itemize}
\item For each given origin node $o\in\mathcal{L}_{{\rm O}}$, all the OD
pairs in $\mathcal{L}_{{\rm OD}}$ originating from $o$ are identified.
Then a probability distribution over these OD pairs is randomly generated.
This probability distribution defines the fractions of the O-flow
$x_{o}$ for the involved OD flows $s_{od}$. 
\item For each given OD pair $od\in\mathcal{L}_{{\rm OD}}$, all the paths
in $\mathcal{L}_{{\rm path}}$ starting from the node $o$ and ending
at the node $d$ are found. A probability distribution over these
paths is randomly generated, which defines the fractions of the OD
flow for the involved path flows $s_{\bm{p}}^{{\rm path}}$. 
\item The assignment matrices related to both the OD flow and the O-flow
models can be computed via Equations (\ref{eq:PathTrafficAssignment2ODTrafficAssignment})
and (\ref{eq:ODTrafficAssignment2OTrafficAssignment}), respectively.
Note that application of (\ref{eq:ODTrafficAssignment2OTrafficAssignment})
requires the knowledge of OD flow assignment matrix $\bm{A}^{t}$,
which further requires the knowledge of path flow assignment matrix
according to (\ref{eq:PathTrafficAssignment2ODTrafficAssignment}).
The previous sub-step is necessary. 
\end{itemize}
\item Finally, we randomly generate O-flows and compute the corresponding
OD flows and link flows. We generate O-flows $\bm{x}^{t}$, $t=2-\tau_{\max},3-\tau_{\max},\cdots,n_{T}$,\footnote{As explained in Remark \ref{rem:time-horizon-Oflows}, $\bm{x}^{t}$,
$t=2-\tau_{\max},3-\tau_{\max},\cdots,n_{T}$, are involved in generating
$\bm{y}^{t}$, $t=1,\cdots,n_{T}$. } according to the transform domain sparsity model described in Section
\ref{subsec:Uniqueness-Unidirectional-Networks} (DCT transform is
used). Based on the generated O-flows, the OD flows $\bm{s}^{t}$,
$t=2-\tau_{\max},\cdots,n_{T}$, can be computed via (\ref{eq:OFlow2ODFlow-multiple}),
and the link flows $\bm{y}^{t}$, $t=1,2,\cdots,n_{T},$ can be evaluated
via (\ref{eq:O2Link-MultiStep}). 
\end{itemize}
\vspace{0cm}

After generating the ground truth data, the test stage proceeds as
follows. Given the link flows $\bm{y}^{t}$, the alternating minimization
procedure described in Section \ref{sec:Blind-Estimation} is used
to estimate the O-flows $\{\hat{\bm{x}}^{t}\}$, $t=1,2,\cdots,n_{T}$,
and the assignment matrices $\{\hat{\bm{P}}^{t}\}$, $t=1,2,\cdots,\tau_{\max}$.

The stopping criteria to quit the iterations include the following: 
\begin{enumerate}
\item The maximum number of iterations is 5000. 
\item Define the normalized mean square error by 
\[
{\rm NMSE}_{y}=\frac{\sum_{t}\left\Vert \hat{\bm{y}}^{t}-\bm{y}^{t}\right\Vert _{2}^{2}}{\sum_{t}\left\Vert \bm{y}^{t}\right\Vert _{2}^{2}}=\frac{\sum_{t}\left\Vert \hat{\bm{P}}^{t}*\hat{\bm{x}}^{t}-\bm{y}^{t}\right\Vert _{2}^{2}}{\sum_{t}\left\Vert \bm{y}^{t}\right\Vert _{2}^{2}}.
\]
The iteration stops when ${\rm NMSE}_{y}<10^{-5}$. 
\item When the sparsity constraint is included, i.e., the formulation (\ref{eq:L1-x-minimization-constrained})
is concerned, another stopping criterion is the speed of convergence.
In particular, we define the decrease of the objective function at
the $k$-th iteration as $\Delta=(\sum_{o=1}^{n_{n}}\left\Vert \bm{D}\bm{x}_{o}\right\Vert _{1})_{k-1}-(\sum_{o=1}^{n_{n}}\left\Vert \bm{D}\bm{x}_{o}\right\Vert _{1})_{k}$.
The iteration stops when $\Delta<10^{-5}$. 
\end{enumerate}
\vspace{0cm}

After obtaining the estimated O-flows $\{\hat{\bm{x}}^{t}\}$ and
the corresponding assignment matrices $\{\hat{\bm{P}}^{t}\}$, we
compute the estimated OD flows $\{\hat{\bm{s}}^{t}\}$, $t=1,2,\cdots,n_{T}$,
via Equation (\ref{eq:OFlow2ODFlow-multiple}). The estimation performance
is measured by the relative error in the estimated OD flows defined
by 
\begin{equation}
E_{s_{od}}^{t}=\frac{\hat{s}_{od}^{t}-s_{od}^{t}}{s_{od}^{t}}.\label{eq:OD-flow-relative-error}
\end{equation}
When the value of $s_{od}^{t}$ is small, a small error in $\hat{s}_{od}^{t}$
may produce large relative error. If an approach performs well in
relative error, it must be good. 

In the simulations, the time horizon $n_{T}$ is chosen to not break
the necessary condition in Theorem \ref{thm:Uniqueness}. The number
of random trials is chosen according to the computation time (decided
by the size of the problem as well as the convergence rate of Algorithm
\ref{alg:AlternativeOpt}). 
\begin{itemize}
\item The 3-by-3 unidirectional network in Figure \ref{fig:3-by-3-Uni-network}:
$n_{T}=60$, and 10 random trials. 
\item The 3-by-3 bidirectional network in Figure \ref{fig:3-by-3-Bi-network}:
$n_{T}=60$, and 100 random trials. 
\item The 8-by-8 bidirectional network in Figure \ref{fig:8-by-8-network}:
$n_{T}=150$, and 10 random trials.
\item The G\'EANT network \cite{uhlig2006providing} in Figure \ref{fig:GEANT-network}:
$n_{T}=150$, and 10 random trials.
\end{itemize}

\subsection{\label{subsec:Simulation-Results}Blind Estimation Made Possible}

\begin{figure}[h]
\begin{centering}
\subfloat[\label{fig:Result_of_3-by-3-Bi_network}The 3-by-3 bidirectional network]{\begin{centering}
\includegraphics[scale=0.5]{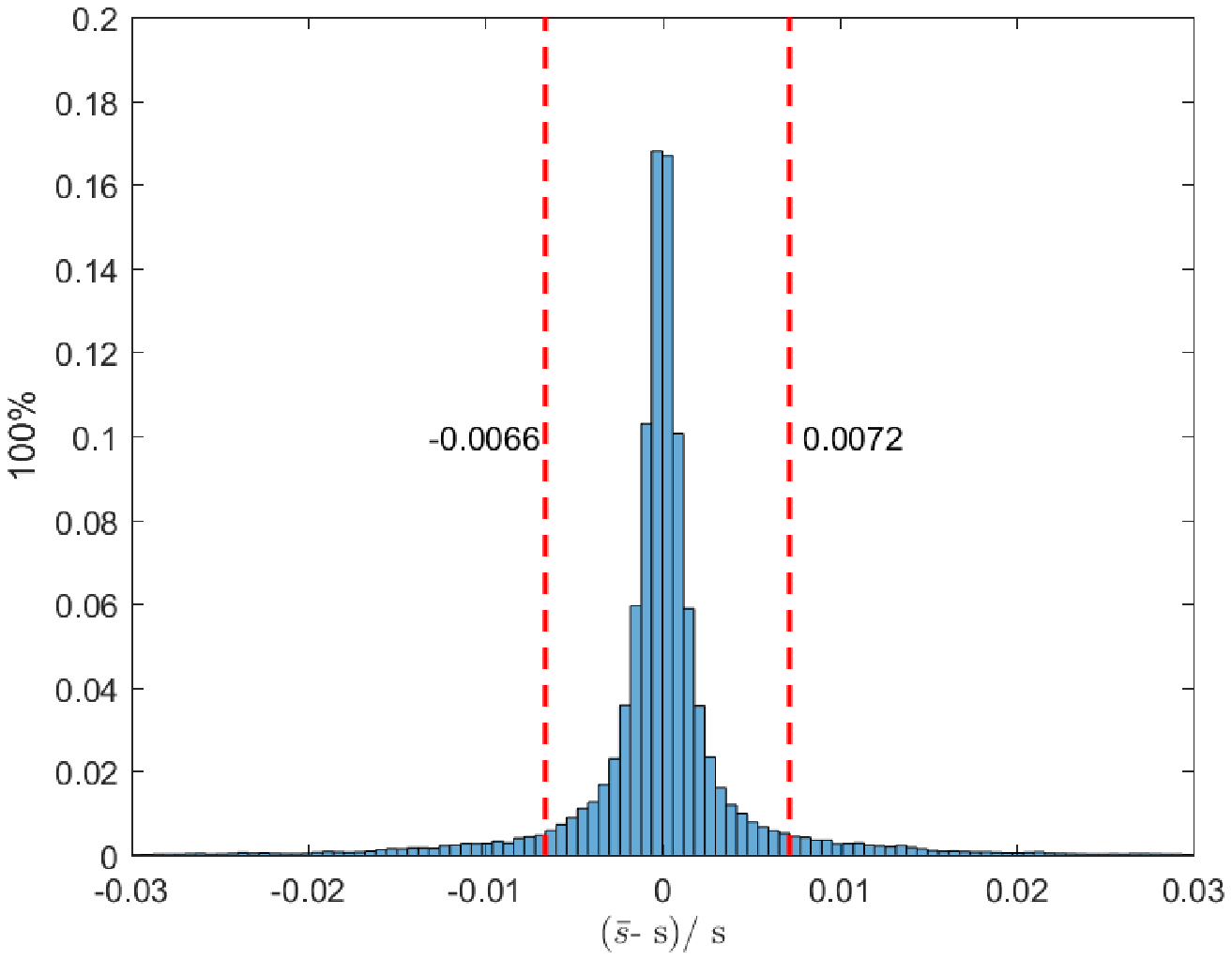} 
\par\end{centering}
}\subfloat[\label{fig:Result_of_8-by-8-Bi_network}The 8-by-8 bidirectional network]{\begin{centering}
\includegraphics[scale=0.5]{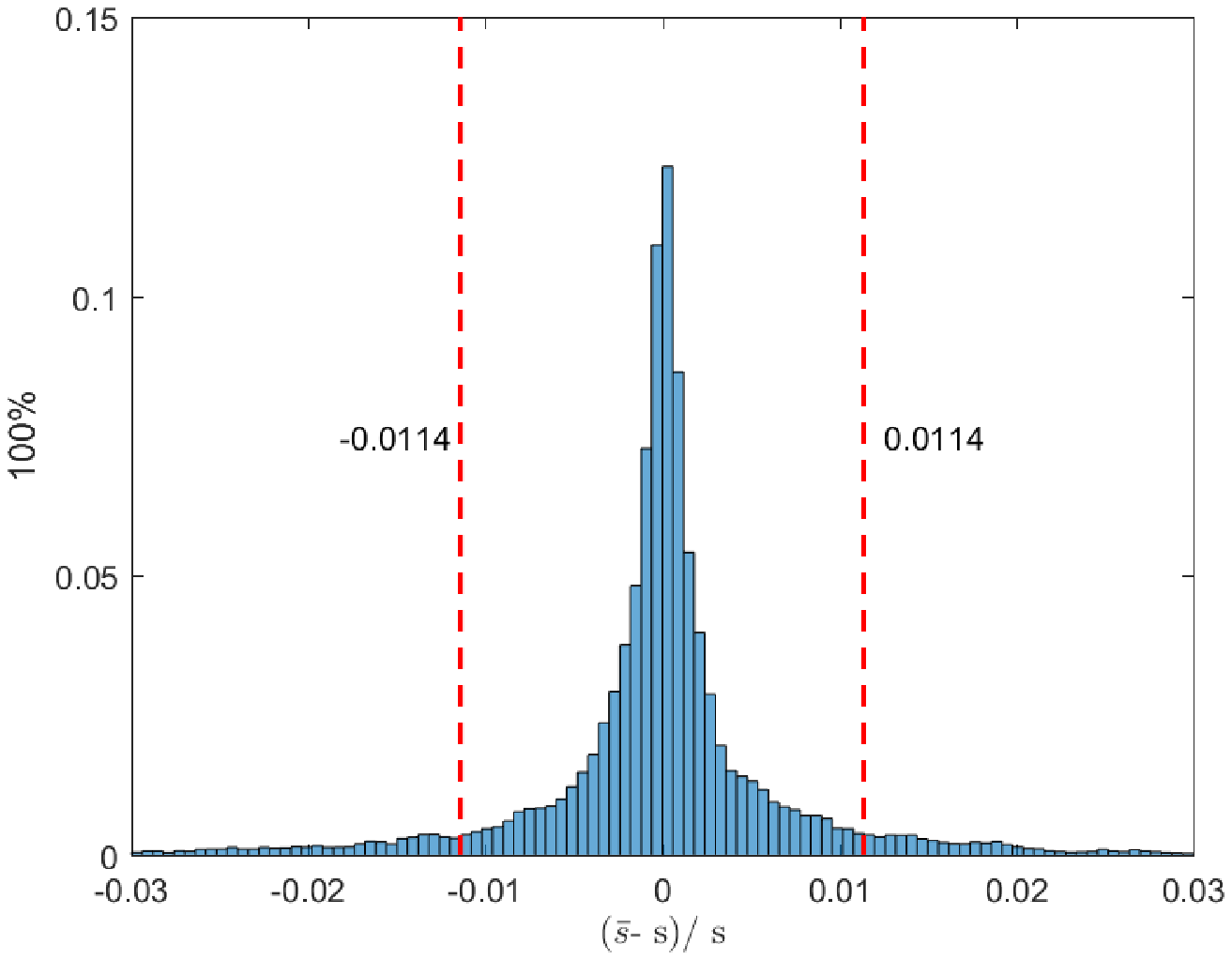} 
\par\end{centering}
}
\par\end{centering}
\caption{\label{fig:Simulation-Bidirectional-Network-1}Histogram of relative
errors: The 3-by-3 and 8-by-8 bidirectional networks }
\end{figure}

\begin{figure}

\begin{centering}
\includegraphics[scale=0.5]{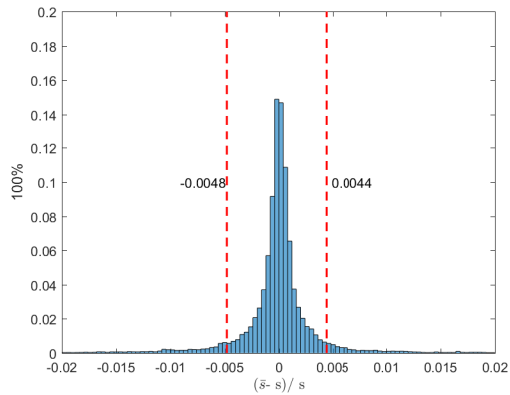}
\par\end{centering}
\caption{\label{fig:Simulation-Bidirectional-Network-2}Histogram of relative
errors: The G\'EANT bidirectional network }

\end{figure}

\begin{figure}[H]
\begin{centering}
\subfloat[\label{fig:UniDirect-WithoutSparsity}Without transform domain sparsity
assumption]{\begin{centering}
\includegraphics[scale=0.5]{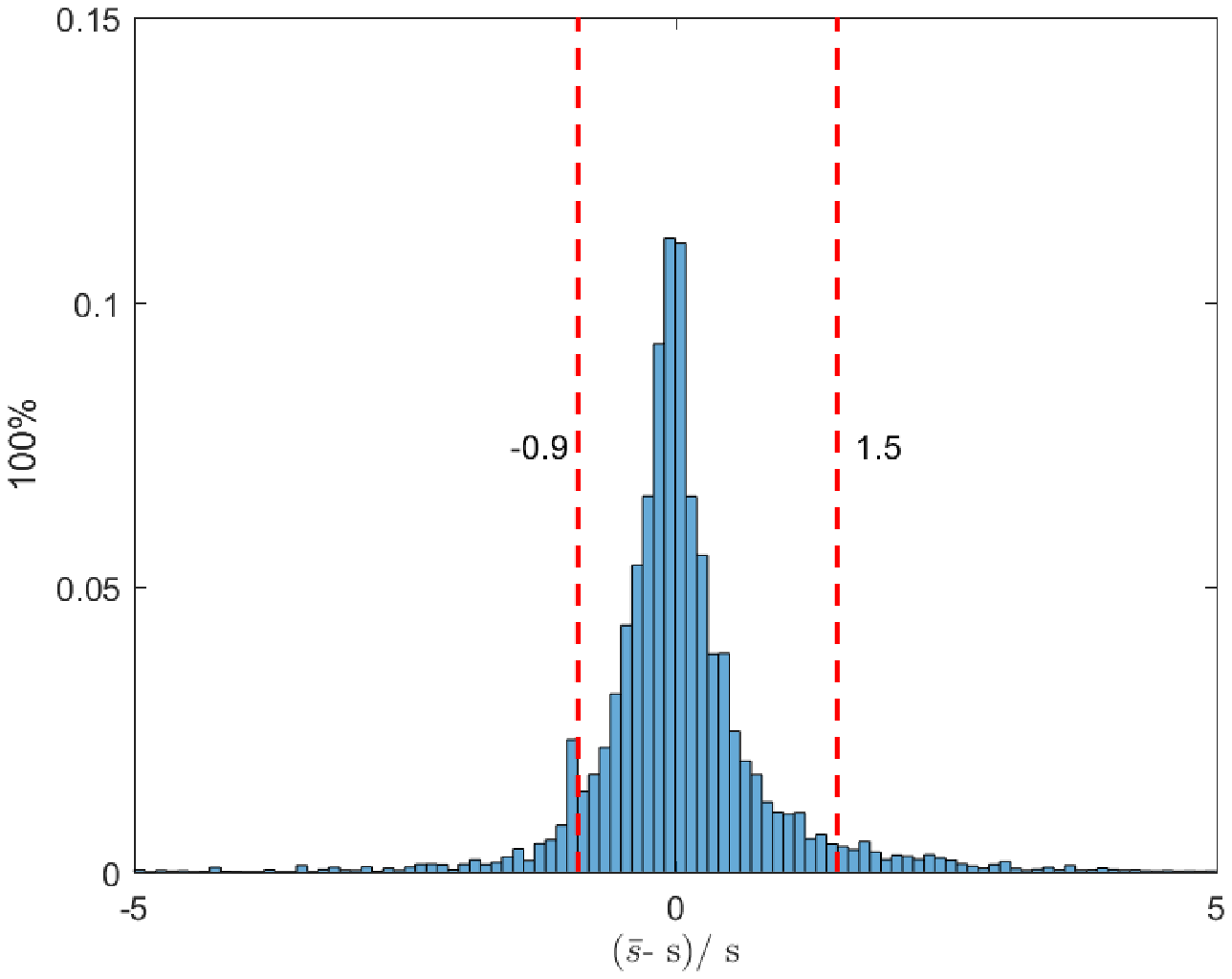} 
\par\end{centering}
}\subfloat[\label{fig:UniDirect-WithSparsity}With transform domain sparsity
assumption]{\begin{centering}
\includegraphics[scale=0.5]{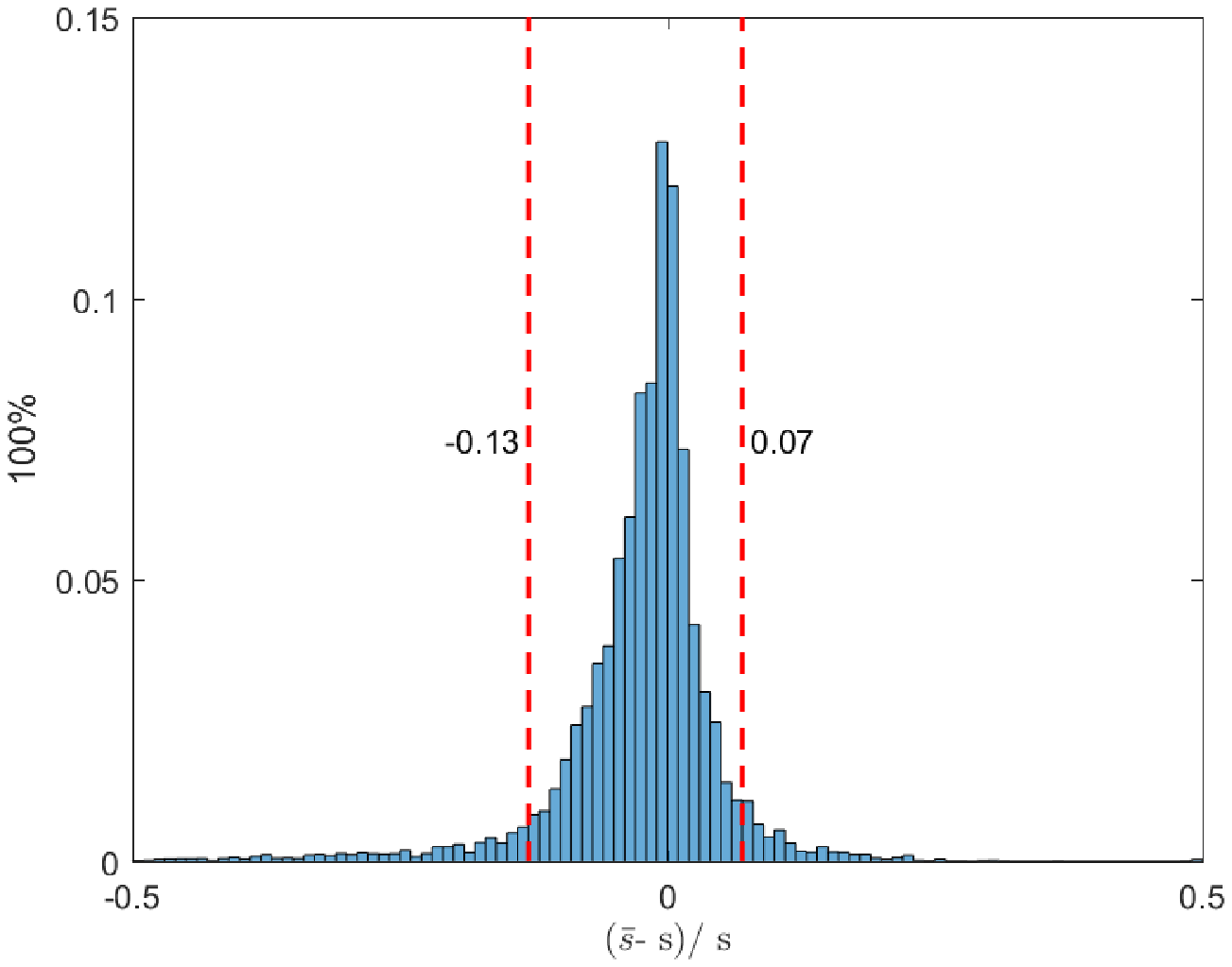} 
\par\end{centering}
}
\par\end{centering}
\caption{\label{fig:Simulation-Unidirectional-Network}Histogram of relative
error: The 3-by-3 unidirectional network. }
\end{figure}

Simulations demonstrate that for bidirectional networks our approach
accurately recovers the OD flows without any prior information. The
simulation results are presented in Figures \ref{fig:Simulation-Bidirectional-Network-1}
and \ref{fig:Simulation-Bidirectional-Network-2} in the form of the
histogram of relative errors (\ref{eq:OD-flow-relative-error}). The
x axis denotes the value of the relative error in the estimated OD
flows. The y axis gives the percentage of OD flows for a given range
of relative error. Each figure also contains two vertical dashed lines:
accumulated 2.5\% of the OD flows have (negative) relative error less
than the value indicated by the left dashed line, accumulated 2.5\%
of the OD flows have relative error larger than the value indicated
by the right dashed line, and accumulated 95\% of the OD flows have
the relative error between the values indicated by the dashed lines.
Simulation results show that 
\begin{itemize}
\item The 3-by-3 bidirectional network: Average relative error (in absolute
value) is less than $0.1\%$. More than 95\% of the estimated OD flows
have relative errors in $\left[-0.66\%,0.72\%\right]$. 
\item The 8-by-8 bidirectional network: Average relative error (in absolute
value) is less than $0.1\%$. More than 95\% of the estimated OD flows
have relative errors in $\left[-1.14\%,1.14\%\right]$. 
\item The G\'EANT network: Average relative error (in absolute value) is
less than $0.1\%$. More than 95\% of the estimated OD flows have
relative errors in $\left[-0.48\%,0.44\%\right]$. 
\end{itemize}
The estimated OD flows are very accurate. 

Simulation results for the 3-by-3 unidirectional network are presented
in Figure \ref{fig:Simulation-Unidirectional-Network}. Figure \ref{fig:UniDirect-WithoutSparsity}
and Figure \ref{fig:UniDirect-WithSparsity} show the histogram of
relative errors before and after considering the transform domain
sparsity, respectively:
\begin{itemize}
\item Before considering transform domain sparsity: Average relative error
(in absolute value) is about $53\%$. More than 95\% of the estimated
OD flows have relative errors in $\left[-90\%,150\%\right]$. 
\item After considering transform domain sparsity: Average relative error
(in absolute value) is about $3\%$. More than 95\% of the estimated
OD flows have relative errors in $\left[-13\%,7\%\right]$. 
\end{itemize}
\vspace{0cm}

The results for the unidirectional network are less impressive compare
to those for bidirectional networks. We suspect that this is due to
the slow convergence of the Gauss-Seidel method, rather than our framework
or the optimization formulation (\ref{eq:L1-x-minimization-constrained}).
Algorithm \ref{alg:AlternativeOpt} terminates when either the number
of iterations is already large or the improvement of the objective
function in (\ref{eq:L1-x-minimization-constrained}) is too small
across adjacent iterations. However, we observe that, in all tested
trials after the termination of the algorithm, the objective function
still decreases along the line linking the output $\left\{ \hat{\bm{P}}^{t}\right\} $
and the ground truth $\left\{ \bm{P}^{t}\right\} $. This observation
suggests that the output solution of Algorithm \ref{alg:AlternativeOpt}
is not a local minimal. An algorithm that solves (\ref{eq:L1-x-minimization-constrained})
with faster convergence may significantly improve the estimation accuracy
for unidirectional networks. 

\section{\label{sec:Conclusion}Conclusion and future work}

To handle the ill-posedness of the OD flow estimation problem, this
paper develops a linear forward model based on the O-flows. The dimension
of the model is substantially reduced, and the OD flow information
is preserved. Simulations demonstrate that for the first time blind
estimation is possible. For bidirectional networks, the ground truth
OD flows can be uniquely identified without any prior information.
A necessary condition for the uniqueness of the solution is derived,
which leads to the conclusion that unidirectional networks in general
do not admit a unique solution under the O-flow model. Nevertheless,
with the assumption of transform domain sparsity, the ground truth
OD flows can be estimated in a reasonable accuracy. 

As a starting point, this paper focuses on relatively simple settings.
It will be beneficial to consider nonlinear traffic models, adapt
the algorithm with different types of prior information, and experiment
with large networks and real data. Furthermore, the algorithmic approach
and Matlab implementation are not optimized, resulting in slow running
speed which makes the current implementation not applicable to large
network/data analysis. Efficient algorithm designs and implementations
can benefit future research. 

\bibliographystyle{IEEEtran}
\bibliography{OD_reference}

\end{document}